\documentclass[graybox, envcountchap]{svmult}

\usepackage{mathptmx}        
\usepackage{amsmath}
\usepackage{amssymb}
\usepackage{color}
\usepackage{helvet}          
\usepackage{courier}         
\usepackage{dirtree}

\usepackage{makeidx}         
\usepackage{graphicx}        
\usepackage{subfig}

\usepackage{multicol}        
\usepackage[bottom]{footmisc}

\usepackage{hyperref}        
\hypersetup{colorlinks=true,urlcolor=blue}

\usepackage[misc]{ifsym}

\makeindex             

\usepackage{amsfonts}
\usepackage{latexsym}
\usepackage{mathrsfs}  
\usepackage{indentfirst}
\usepackage{cite}

\newtheorem{observation}[remark]{Observation}

\newcommand{\beq}{\begin{eqnarray}}
\newcommand{\eeq}{\end{eqnarray}}
\newcommand{\pa}{\partial}

\newcommand{\diag}{\,\mbox{diag}\,}

\renewcommand{\Re}{\,\mbox{Re}\,}

\newcommand{\eq}[1]{(\ref{#1})}
\newcommand{\n}[1]{\label{#1}}

\newcommand{\al}{\alpha}
\newcommand{\be}{\beta}

\newcommand\Ga{\Gamma}

\newcommand\de{\delta}
\newcommand\De{\Delta}

\renewcommand\th{\theta}

\newcommand\la{\lambda}
\newcommand\La{\Lambda}

\newcommand\rh{\rho}
\newcommand\si{\sigma}

\newcommand\ph{\varphi}
\newcommand\om{\omega}

\newcommand\lap{\Delta}
\newcommand\na{\nabla}
\renewcommand\pa{\partial}

\newcommand{\rd}{\mathrm{d}}  

\renewcommand\thechapter{4}

\begin{document}


\title{Regular black holes from higher-derivative effective delta sources$^*$}
\titlerunning{Regular black holes from higher-derivative effective delta sources}

\author{Breno L. Giacchini and Tib\'{e}rio de Paula Netto}

\institute{Breno L. Giacchini (\Letter) \at Department of Physics, Southern University of Science and Technology, Shenzhen 518055, China, \email{breno@sustech.edu.cn}
\and Tib\'{e}rio de Paula Netto (\Letter) \at Department of Physics, Southern University of Science and Technology, Shenzhen 518055, China \email{tiberio@sustech.edu.cn}
\\
\vspace{0.2cm}
\\
$^*$ This is a preprint of the following chapter: Breno L. Giacchini and Tib\'{e}rio de Paula Netto, Regular Black Holes from Higher-Derivative Effective Delta Sources, published in Regular Black Holes: Towards a New Paradigm of Gravitational
Collapse, edited
by Cosimo Bambi, 2023, Springer, reproduced with
permission of Springer Nature Singapore Pte. The final authenticated version is
available online at: \url{http://doi.org/10.1007/978-981-99-1596-5_4}.
}
\maketitle

\vspace{-1.3cm}

\abstract{
Certain approaches to quantum gravity and classical modified gravity theories result in effective field equations in which the original source is substituted by an effective one. In these cases, the occurrence of regular spacetime configurations may be related to the regularity of the effective source, regardless of the specific mechanism behind the regularization. In this chapter, we make an introduction to the effective source formalism applied to higher-derivative gravity. The results presented here, however, can be easily transposed to other frameworks that use similar sources. The generality obtained is also because we consider a general higher-derivative gravity model instead of restricting the analysis to some specific theories. In the first part, we discuss  the model in the Newtonian limit, which offers a natural context for introducing effective sources. We show how the regularity properties of the effective sources depend on the behavior of the action's form factor in the ultraviolet regime, which leads to results valid for large families of models (or for families of modified delta sources). Subsequently, we use the general results on the effective sources to construct regular black hole metrics. One of our concerns is the higher-order regularity of the solutions, \textit{i.e.}, the possibility that not only the invariants built with curvature tensors but also the ones with covariant derivatives are regular. In this regard, we present some theorems relating the regularity of sets of curvature-derivative invariants with the regularity properties of the effective sources.
}


\section{Introduction}
\label{Sec1}

It is widely expected that quantum effects would resolve the spacetime singularities present in classical 
solutions of general relativity. The process behind this regularization, nevertheless, is not entirely understood. It might ultimately rely on a consistent quantum description of gravity, which remains an open problem.

In many approaches to quantum gravity, the regularization mechanism can be described, usually approximately or effectively, via the replacement of pointlike localized structures with smeared objects. In this 
perspective, the effective source formalism is a useful framework for studying the resultant spacetime configurations. It has been applied, for example, in the context of noncommutative geometry, generalized uncertainty principle models, string theory, and higher-derivative gravity (see, \textit{e.g.},~\cite{dirty,Nicolini:2005vd,Isi:2013cxa,Tseytlin95,Modesto:2010uh,Bambi-LWBH,Jens,Nos6der,BreTib2,Zhang14,Modesto12,Nicolini:2019irw} and references therein).

Our goal in this chapter is to introduce the effective source formalism and show some regularity properties of black holes constructed using such sources. We choose to work in the framework of higher-derivative gravity.  The motivation for this choice is not only pedagogical but also because many interesting and general results about these models  were obtained first using this formalism.

We emphasize, however, that the results presented here are very general and can be easily transposed to other frameworks.  The generality obtained is a consequence of the wide scope in which we consider higher-derivative gravity. Even though  this term has been more frequently associated with the Einstein-Hilbert action augmented by 
fourth-derivative structures, such as $R^2$ and $R_{\mu\nu}^2$, it 
can be applied to any extension of general relativity whose action 
contains more than two derivatives of the metric. In an even broader 
sense, it refers to gravitational theories whose propagator displays an improved behavior in the ultraviolet (UV) regime, as it happens, \textit{e.g.}, in some families of nonlocal gravity models defined by actions that are nonpolynomial in derivatives of the metric.

Solving the equations of motion in higher-derivative gravity is a highly complex task;
in the case of static spherically symmetric metrics 
only in fourth-derivative gravity the space of solutions of the nonlinear field equations has been studied rigorously~\cite{Stelle78,Stelle15PRD,Lu:2015cqa,Cai:2015fia,Holdom:2002xy,Holdom:2016nek,Holdom:2022zzo,Podolsky:2018pfe,Svarc:2018coe,Bonanno:2019rsq,Bonanno:2021zoy}. Regarding theories with six or more derivatives, the solutions obtained usually involve certain approximations of the field equations, such as the Newtonian or the small-curvature approximation. The effective source formalism proves to be a useful tool in both situations, as it enables the derivation of general results on the regularity of the solution without the need to solve all the equations for a specific model. As we show here, some features of the solution only depend on the behavior of the higher-derivative sector of the model, and this connection is made explicit via the effective source. Indeed, if the effect of the higher derivatives can be interpreted as the regularization of a singular source, one might study the effective source to investigate the regularity of the associated spacetime configuration.

Nevertheless, what is expected of a ``regular'' metric depends on the application under consideration. While the regularity of the metric components in a specific coordinate chart can suffice to define a bounded (modified) Newtonian potential, it might not be enough to avoid the divergence of components of the Riemann curvature tensor. Thinking of general covariance, a more useful definition of regular spacetime relies on the regularity of an invariant (or a particular set of invariants). In many cases, the Kretschmann scalar is chosen, but in other applications it might be necessary also to consider scalars involving covariant derivatives (see, \textit{e.g.}, \cite{Borissova:2020knn,Giacchini:2021pmr}). 
A simple example showing that the regularity  of the Kretschmann scalar does not imply that all curvature-derivative invariants are bounded is provided by the Hayward metric~\cite{Hayward},
\beq \label{Met_Hayward}
\rd s^2 = - \left( 1 - \frac{ 2  M  r^2}{r^3 + 2L^3} \right)  \rd t^2 + \left( 1 - \frac{ 2  M  r^2}{r^3 + 2L^3} \right)^{-1} \rd r^2 + r^2 \rd \Omega^2,
\eeq
where $\rd\Omega^2$ is the metric of the unit two-sphere. It is straightforward to verify that while $R$, $R_{\mu\nu}^2$, $R_{\mu\nu\al\be}^2$, and even the scalars cubic in curvatures are all bounded, invariants containing derivatives of curvatures might be singular; for instance,
\beq \label{Box2R_H}
\Box^2 R \underset{r \to 0}{\sim} - \frac{504  M}{L^6 r} .
\eeq

In this regard, the effective source formalism, as presented here,
is helpful in the classification of the order of regularity of certain families
of black hole solutions, according to the behavior of the effective source. 
Moreover, since field equations with similar effective sources are also obtained in different frameworks (even from other approaches to quantum gravity), these results can be applied in a broader context.

This chapter is organized as follows. In Sec.~\ref{Sec2}, we make a brief 
review about the importance of higher-derivative terms in the gravitational action, mainly motivated by quantum considerations. The effective source is introduced in Sec.~\ref{Sec3}, where we analyze the Newtonian limit of a generic higher-derivative gravity model. Important results on the effective sources are proved in Sec.~\ref{Sec4}, followed by a discussion on the definition of higher-order regularity of the metric, which is a crucial notion in the study of the regularity of curvature invariants with covariant derivatives. Then, in Sec.~\ref{Sec6} we present some explicit examples of the general results obtained in the preceding sections by focusing on particular models, namely, polynomial and  nonlocal gravity. In Sec.~\ref{Sec7} we construct regular black hole solutions based on the results and examples derived in the other sections. Finally, our concluding remarks are given in Sec.~\ref{Conc}.

Throughout this chapter, we use the {mostly plus metric convention}, with the
Minkowski metric $\eta_{\mu\nu} = \diag (-1,+1,+1,+1)$.
The Riemann curvature tensor is 
defined by
\beq
\n{Rie}
R^\al {}_{\be\mu\nu} \,=\, 
\pa_\mu \Ga^\al_{\be \nu} - \pa_\nu \Ga^\al_{\be \mu} 
+ \Ga^ \al_{\mu \tau} \, \Ga^\tau_{\be \nu}  
- \Ga^\al_{\nu \tau } \, \Ga^\tau _{\be \mu} 
\, ,
\eeq
while the Ricci tensor and the scalar curvature are, respectively,  
$R_{\mu\nu} = R^\al {}_{\mu\al\nu}$ and $R = g^{\mu\nu} R_{\mu\nu}$.
Also, we use the notation
\beq
G_{\mu\nu} = R_{\mu\nu} - \frac12 g_{\mu\nu} R
\eeq
for the Einstein tensor, and we adopt the unit system such that $c = 1$ and $\hslash  = 1$.


\section{Higher-derivative gravity models}
\label{HD sec}
\label{Sec2}

The idea of generalizing the Einstein-Hilbert action by including higher-order terms can be traced back to the early years of general relativity (see, \textit{e.g.},~\cite{Eddington}). Nevertheless, only in the last decades these models have attracted more attention, and the main motivation for this comes from the quantum theory.

In fact, since the 1970s, it has been known that the quantum version of general relativity, based on the Einstein-Hilbert action (with or without the cosmological constant $\La$), 
\beq
\n{SEH}
S_{\rm EH} = \frac{1}{16 \pi G} \int \rd^4 x \sqrt{-g} ( R + 2 \La ),
\eeq
is not perturbatively renormalizable~\cite{hove,dene,GorSag,Christensen:1979iy}. This is mainly due to the 
negative mass dimension of the Newton constant $G$. However, if the gravitational action is enlarged by fourth-derivative terms,
\beq
\n{4HD}
S_{\rm grav} = S_{\text{EH}} +  \int \rd^4 x \sqrt{-g} \left\{
\al_1 R_{\mu\nu\al\be}^2 + \al_2 R_{\mu\nu}^2 + \al_3 R^2 +\al_4 \Box R 
\right\},
\eeq 
where $\Box = g^{\mu\nu} \na_\mu \na_\nu$ is the d'Alembert operator,
the corresponding quantum model turns out to be renormalizable~\cite{Stelle77}. It is interesting to notice that the 
structures in~\eq{4HD} also appear in the semiclassical theory in order to renormalize the vacuum fluctuations of quantum fields in a classical curved background (see, \textit{e.g.}, \cite{birdav,book} for an introduction).

One of the most severe drawbacks of higher-derivative models is the ghostlike particles that typically exist in the theory's spectrum. From the classical point of view, such modes are characterized by having negative kinetic energy,  while at the quantum level, they usually cause violation of unitarity.\footnote{Ghosts should not be mistaken for tachyons, which are instead characterized by negative mass squared.} 
In recent years, however, considerable effort has been made to understand the ambiguous role of higher-derivative ghosts in perturbative quantum gravity. 
This led to a multiplicity of insights on how to tame (or avoid) the ghosts and restore unitarity\footnote{As the technical details regarding unitarity in these models lie beyond the scope of this text, we  refer the interested reader to the original references mentioned above for further consideration.} \cite{AsoreyLopezShapiro,ModestoShapiro16,Modesto16,AnselmiPiva2,Anselmi:2017ygm,Bender:2007wu,Bender:2008gh,Donoghue:2019fcb,Krasnikov,Kuzmin,Tomboulis,Modesto12,Biswas:2011ar}. Among these proposals, we mention the cases of Lee-Wick gravity~\cite{ModestoShapiro16,Modesto16}, which is based on a formulation of polynomial-derivative gravity with  ghosts with complex masses, 
and nonlocal ghost-free gravity~\cite{Krasnikov,Kuzmin,Tomboulis,Modesto12,Biswas:2011ar}, based on actions that are nonpolynomial in derivatives of the metric.

In the former case, one introduces in the action~\eq{4HD} operators that are polynomial in the d'Alembertian, namely,
\begin{eqnarray}
\n{ALS}
\hspace{-0.5 cm}  
\,\sum_{n=1}^{N-1} \int \rd^4 x \sqrt{-g} \,
\Big\{
\al_{1,n} \, R_{\mu\nu\al\be}\, \square^n\,R^{\mu\nu\al\be}
+ \, \al_{2,n} \, R_{\mu\nu} \, \square^{n} R^{\mu\nu}
+ \, \al_{3,n} \, R \, \square^{n} R
\Big\} 
\,.
\label{action-high}
\end{eqnarray}
Since each curvature contains two metric derivatives, the resulting action contains $2N+2$ derivatives. 
Compared with the fourth-derivative gravity~\eq{4HD}, which can be at most strictly renormalizable, the presence of derivatives higher than four can render this theory superrenormalizable~\cite{AsoreyLopezShapiro}. In the particular case of Lee-Wick gravity, the coefficients $\al_{\ell,n}$ are chosen 
so that all the ghostlike modes correspond to complex conjugate pairs of poles in the propagator.

In addition to the terms~\eq{ALS}, 
it is possible to include 
structures of higher order in curvatures without spoiling the superrenormalizability of the model. This happens provided that such terms contain at most $2N+2$ derivatives of the metric, for example,  
\beq
\n{kkkil}
R^2 \Box^{N-3} R^2 \qquad \text{and} \qquad R_{\mu\nu} R^{\mu\nu} \Box^{N-3}  R_{\al\be} R^{\al\be} \,.
\eeq
Terms of this type can even be beneficial, as with a judicious choice of their front coefficients, they can be used to cancel the loop divergences and make the theory finite at the quantum level~\cite{AsoreyLopezShapiro,Modesto16,Modesto:2014lga}.

On the other hand, in the case of nonlocal ghost-free higher-derivative gravity models, nonlocalities are introduced at the classical level in such a way that the UV behavior of the graviton propagator gets modified without introducing new (ghost) degrees of freedom. The simplest model with these characteristics is described by the action
\beq
\n{Snl}
S_{\rm grav} = S_{\text{EH}} \, + \, \frac{1}{16 \pi G} \int \rd^4 x \sqrt{-g} \, 
G_{\mu\nu} \frac{e^{H(\Box)} - 1}{\Box} R^{\mu\nu}
\,,
\eeq
where $H(z)$ is an entire function. 
Besides being free of ghosts, for some choices of the function $H(z)$
the model~\eq{Snl} can be (super)renormalizable. The reader can consult~\cite{Krasnikov,Kuzmin,Tomboulis,Modesto12} for the details.


\section{Higher-derivative gravity in the Newtonian limit}
\label{weak-field}
\label{Sec3}

In the previous section we presented some families of higher-derivative gravity models whose main motivations come from the quantum field theory point of view.
Here we start to consider these models in the classical domain to discuss whether (and how) they can lead to regular spacetime configurations. It is sound to begin with the simplest singularity in gravity, related to the Newtonian singularity proportional to $1/r$.
In this spirit, in this section, we discuss in detail the Newtonian limit of a general higher-derivative gravity model.


\subsection{Linearized higher-derivative gravity}

Instead of considering each model described in Sec.~\ref{HD sec} separately, it is possible to develop a general formalism for linearized higher-derivative gravity. This comes from the observation that, in the linear approximation, any higher-derivative gravity can be expressed in the form of the action\footnote{In the action~\eq{HDF}, we did not include total derivative terms since they do not contribute to the equations of motion, nor did we include the cosmological constant because it is irrelevant for metric perturbations around the Minkowski spacetime.}
\beq
\n{HDF}
S_{\rm grav} = \frac{1}{16 \pi G} \int \rd^4 x \sqrt{-g} \, \left\{ R + 
R F_1 (\Box) R  +  R_{\mu\nu} F_2(\Box) R^{\mu\nu} \right\} 
,
\eeq
where $F_{1,2}$ are arbitrary functions, called form factors, that depend on the specific model under consideration.

In fact, in the linear regime, we consider metric fluctuations $h_{\mu\nu}$ around the Minkowski spacetime,
\beq
\n{flat-exp}
g_{\mu\nu} = \eta_{\mu\nu} + h_{\mu\nu}
,
\qquad
|h_{\mu\nu}| \ll 1 ,
\eeq
and we are interested in equations of motion
\beq
\n{EOM-def}
\frac{\de S_{\rm grav}}{\de h_{\mu\nu}} 
\,
\eeq
that are  
in the first order in $h_{\mu\nu}$. Since the linear terms in~\eq{EOM-def} are originated from the second-order terms in the action, the only relevant higher-derivative structures in the linear regime are the ones capable of generating quadratic terms in $h_{\mu\nu}$. Expanding the curvatures in powers of the perturbation $h_{\mu\nu}$ (see the Appendix for the explicit formulas), we notice that the Riemann curvature tensor is already $O(h_{\ldots})$ because the Minkowski background is flat. 
From this trivial observation, we can derive important conclusions about the types of terms that contribute to the linear approximation:
\begin{enumerate}
\item{
Terms of order higher than two in curvatures are irrelevant in this limit, since they are $O(h_{\ldots}^3)$. This is the case, \textit{e.g.}, of the terms in~\eq{kkkil}.
} 
\item{
Every term quadratic in curvature is already 
$O(h_{\ldots}^2)$. 
Thus, if such terms contain covariant derivatives, 
only the zero-order term of the expansion of the derivative can contribute in the linear limit. So, in this case,
we can simply trade $\na_\mu \mapsto \pa_\mu$ for all derivatives in curvature-squared terms. 
Therefore, by changing the order of the derivatives, applying integration by parts, ignoring surface terms, and using the Bianchi identities, it is possible to prove that any quadratic curvature term in the action can be reduced to three structures, namely,
\beq
\n{cu-esc}
R F_1 (\Box) R
\,, 
\quad  
R_{\mu\nu} F_2(\Box) R^{\mu\nu} \, ,
\quad \mbox{and} \quad
R_{\mu\nu\al\be} F_3(\Box) R^{\mu\nu\al\be}
\,.
\eeq 
} 
\item{Among these three structures, only two 
are linearly independent at order $h_{\ldots}^2$. Indeed, using the formulas~\eq{R-exp1}--\eq{R-exp3} of the Appendix, it is immediate to prove that, for any $F (\Box)$,
\beq
\n{GGB-n}
\left[ R_{\mu\nu\al\be} F (\Box) R^{\mu\nu\al\be} - 4 R_{\mu\nu} F (\Box) R^{\mu\nu} + R F (\Box) R \right]^{(2)} 
= 0 
, 
\eeq
where the superscript indicates that the equation is valid at order $h_{\ldots}^2$. 
Hence, applying~\eq{GGB-n}, it is possible to eliminate one of the curvature scalars in
the basis~\eq{cu-esc}. For example, if we start with the action
\beq
\n{HDF2}
S_{\rm grav} &=& \frac{1}{16 \pi G} \int \rd^4 x \sqrt{-g} \, \big\{R + 
R \tilde{F}_1 (\Box) R  +  R_{\mu\nu} \tilde{F}_2(\Box) R^{\mu\nu} 
\nonumber
\\
&& \hspace{0.9cm}
+ \, R_{\mu\nu\al\be} \tilde{F}_3 (\Box) R^{\mu\nu\al\be}
\big\} 
,
\eeq
with the aid of~\eq{GGB-n} one can reduce it to~\eq{HDF} under the redefinition of the form factors
\beq
\n{F-rel}
F_1 (\Box) = \tilde{F}_1 (\Box) - \tilde{F}_3 (\Box)
\,,
\qquad
F_2 (\Box) = \tilde{F}_2 (\Box) + 4 \tilde{F}_3 (\Box)
.
\eeq
With this procedure, we can eliminate the Riemann-squared term in the action. Note, however, that there is nothing special in removing the Riemann-squared term; according to the problem under consideration, one can use \eq{GGB-n} to rewrite the action on the most convenient basis consisting of any two curvature-quadratic invariants (even including the square of the Weyl tensor).
}
\end{enumerate}

\begin{observation} \label{Obs1}
\normalfont 
\begin{svgraybox}
In some situations, the identity~\eq{GGB-n} can be seen as the linearized version of the
Gauss-Bonnet-like relations that hold in four-dimensional spacetime. 
Indeed, it is well known that the integrand of the topological Gauss-Bonnet term is a total derivative (see, \textit{e.g.},~\cite{hove}),
\beq
\n{GB-basic}
R_{\mu\nu\al\be}^2  - 4 R_{\mu\nu}^2 + R^2  = \text{total derivative.}
\eeq
Although the same cannot be said if we insert any power of the d'Alembertian in between the curvatures in~\eq{GB-basic}, it is possible to show that
\beq
\n{GB-nl}
&&
R_{\mu\nu\al\be} \Box^\ell  R^{\mu\nu\al\be} - 4 R_{\mu\nu} \Box^\ell R^{\mu\nu} + R \Box^\ell R = O(R_{\ldots}^3)
\nonumber
\\
&&  \hspace{1.3cm}
+\, \text{\,total derivative}
\qquad \forall \, \ell \in \lbrace 1,2,3,... \rbrace 
\,.
\eeq
The proof of~\eq{GB-nl} involves commuting covariant derivatives and applying the Bianchi identities, see~\cite{Deser:1986xr,AsoreyLopezShapiro} for the details.

Therefore, if $F(\Box)$ can be expressed as a power series, Eq.~\eq{GB-nl} implies in~\eq{GGB-n}. On the other hand, identity~\eq{GGB-n} holds for any function $F$, even for nonanalytic form factors. This last observation is important because quantum corrections for the gravitational action typically have the form of $\log (-\Box/\mu^2)$. Further discussion on logarithmic quantum corrections to the Newtonian potential can be found in~\cite{Nos4der,Nos6der,dePaulaNetto:2021axj} and references therein.
\end{svgraybox}
\end{observation}

Taking into account the above results, it follows that in the linear approximation, any higher-order term in the gravitational action is equivalent to the structures in~\eq{HDF}. It is essential, however, to keep in mind that this equivalence is valid only at second order in the metric perturbation (\textit{i.e.}, at the level of linearized equations of motion); the same also applies to the freedom mentioned above for expressing the action in any preferred curvature basis. In fact, in applications where terms of higher order in the metric perturbation are relevant, different curvature invariants can give different contributions. For example, this happens at the classical level when one considers the solutions of the complete (nonlinear) field equations. At the quantum level, we mention that the quadratic part of the action is enough to determine the propagator of the model, while to define the structure of vertices, one has to examine the higher-order terms.

Using the formulas in the Appendix, it is possible to show that the $h_{\mu\nu}$-bilinear part of~\eq{HDF} 
is given by
\beq
\n{HD-weak2}
S_{\rm grav} &=& \frac{1}{32 \pi G}  \int \rd^4 x \,  \Big[   
 \frac12 h_{\mu\nu} a_1 (\Box) \Box h^{\mu\nu} 
- \frac12 \, h a_2 (\Box) \Box h
- h^{\mu\nu} a_1 (\Box) \pa_\mu \pa_\la  h^{\la}_{\nu}
\nonumber
\\
&&
+ h a_2 (\Box) \pa_\mu \pa_\nu h^{\mu\nu} 
+
\frac12 h_{\al\be} [ a_1 (\Box) - a_2 (\Box) ]  \frac{\pa^\al \pa^\be \pa^\mu \pa^\nu}{\Box} h_{\mu\nu}
\Big]
\,,
\eeq
where the functions $a_{1,2} (z)$ relate to the form factors $F_{1,2}(z)$ via
\begin{subequations} \label{a1a2f}
\begin{align} 
& \n{a1f}
a_1 (\Box) = 1 + \, F_2 (\Box) \Box
\,, \\
&  \n{a2f}
a_2 (\Box) = 1 - [ 4 F_1 (\Box) + F_2 (\Box) ] \Box 
\,.
\end{align}
\end{subequations}

Finally, the coupling with matter is introduced through 
\beq
S_{\text{m}} =  \frac12 \int \rd x \,\, T^{\mu\nu}  h_{\mu\nu},
\eeq
where $T_{\mu\nu}$ is the energy-momentum tensor of matter, such that the total action is $\,S_{\text{total}} = S_{\text{grav}} + S_{\text{m}}$. The principle of least action,
\beq
\frac{ \de S_{\text{total}}}{\de h_{\mu\nu}} = 0,
\eeq 
results in the linearized equations of motion,
\beq
\n{EOM-GR}
\varepsilon^{\mu\nu} = - 16 \pi G \, T^{\mu\nu}
\,,
\eeq
where 
\beq
\n{EOMHD}
\varepsilon^{\mu\nu} &=&   a_1 (\Box) \big(  \Box h^{\mu\nu} - \pa^\mu \pa^\la h^\nu_\la - \pa^\nu \pa^\la h^\mu_\la \big)  
+ a_2 (\Box)   \big[   \eta^{\mu\nu} \big( \pa^\al \pa^\be h_{\al\be} - \Box h \big)  
\nonumber
\\
&&
+ \pa^\mu \pa^\nu h \big]
+ [ a_1 (\Box) - a_2 (\Box) ]  \, \frac{\pa^\mu \pa^\nu \pa^\al \pa^\be}{\Box} h_{\al\be}
\,.
\eeq


\subsection{Field equations in the Newtonian limit}

The Newtonian limit means that the gravitational field is weak and static. In this situation, we can adopt Cartesian coordinates $x^\mu = (t,x,y,z)$ and write the metric in a generic space-isotropic form
\beq
\n{New-met}
\rd s^2 = -(1 + 2 \ph) \rd t^2 + (1 - 2 \psi) (\rd x^2 + \rd y^2 + \rd z^2) 
\,,
\eeq  
where $\ph$ and $\psi$ ($|\ph|, |\psi| \ll 1$) represent the two independent components of the metric, often called  \emph{Newtonian potentials}.
Moreover, the Newtonian approximation also assumes that the matter is nonrelativistic, described by the energy-momentum tensor 
\beq
T_{\mu\nu} = \de^0_\mu \de^0_\nu \rho
\,,
\eeq
where $\rho$ is the mass density.

To obtain the equations for the potentials it is sufficient to consider the 00-component and the trace of Eq.~\eq{EOM-GR}. Since the metric is static, $\pa_0 h_{\mu\nu} = 0$; and we get, respectively,
\beq
\n{1} &
 \varepsilon_{00}  = a_1 (\Box)  \Box h_{00} + \eta_{00} a_2(\Box) (\pa_\al \pa_\be h^{\al\be} - \Box h ) = - 16\pi G \, \rho
\,,  &
\\
\n{2}
&
\eta^{\mu\nu} \varepsilon_{\mu\nu} =  [3 a_2(\Box) - a_1 (\Box) ] ( \pa_\al \pa_\be h^{\al\be} - \Box h) = - 16\pi G \eta^{00} \rho
\,. &  
\eeq
In our sign convention $\eta^{00} = -1$, $\eta^{ij} = \de^{ij}$ and, consequently, $\Box = \De$ when applied to a static field. 
Inasmuch as $h_{00} = - 2 \ph $ and $h_{ij} = - 2 \psi \de_{ij}$, it follows
\begin{equation}
\n{hh}
\Box h_{00} = - 2 \De \ph\, 
\qquad 
\mbox{and} 
\qquad 
\pa_\al \pa_\be h^{\al\be} - \Box h =   2  ( 2 \De \psi - \De \ph  )
\,.
\end{equation}
Thus, we find for \eq{1} and \eq{2}, respectively, 
\begin{subequations} \label{EOM-Sistema}
\begin{align} 
&[a_1 (\De) - a_2 (\De)] \De \ph + 2 a_2 (\De) \De \psi 
= 8 \pi G \rho,  \label{EOM-00} \\
&[ 3 a_2 (\De) - a_1 (\De) ] ( 2  \De \psi - \De \ph  ) 
= 8 \pi G \rho.  \label{EOM-tr}
\end{align}
\end{subequations}

\begin{proposition}
\label{a1=a2}
If $a_1 (\Box) = a_2 (\Box) \equiv a (\Box)$, the two Newtonian potentials $\ph$ and $\psi$ are equal and determined by the modified Poisson's equation
\beq
\n{mod-lap-1}
a(\De) \De \ph =  4 \pi G \rho
\,. 
\eeq
\end{proposition}
\begin{proof}
Putting $a_1 (\Box) = a_2 (\Box) \equiv a (\Box)$ in the system~\eq{EOM-Sistema} we obtain 
\beq
a (\De) \De \psi 
= 4 \pi G \rho ,
\qquad
a (\De)  (2  \De \psi - \De \ph ) 
= 4 \pi G \rho
\,.
\eeq
Then, using the first equation in the second one, it follows
\beq
a(\De) \De \ph = a(\De) \De \psi =   4 \pi G \rho
\,.
\eeq
Since both potentials are defined by the same equations (and are subjected to the same boundary conditions), we get $\ph = \psi$.
\qed
\end{proof}

\begin{observation} \label{Obs2}
\normalfont 
\begin{svgraybox}
The condition $a_1(\Box) = a_2 (\Box) = a (\Box)$ means that [see Eq.~\eq{a1a2f}]
\beq
\n{F-lm}
F_2 (\Box) = - 2 F_1 (\Box) = \frac{a (\Box) - 1}{\Box}.
\eeq
So, substituting~\eq{F-lm} into the gravitational action~\eq{HDF}, we have
\beq
\n{act-lm}
S_{\rm grav} = \frac{1}{16 \pi G} \int \rd^4 x \sqrt{-g} \Big\{ R + G_{\mu\nu} \frac{a (\Box) - 1}{\Box} R^{\mu\nu} \Big\}.
\eeq
The action~\eq{act-lm} represents the only family of 
higher-derivative models in which the two Newtonian potentials are equal and defined by the differential equation~\eq{mod-lap-1}. Also, general relativity is recovered for $a_1(\Box) = a_2 (\Box) = 1$ or, in an equivalent way, $F_1 (\Box) = F_2 (\Box) = 0$.
\end{svgraybox}
\end{observation}

In the case of more general functions $a_1(\Box)$ and $a_2 (\Box)$, it is convenient to rewrite the system~\eq{EOM-Sistema} in terms of variables that depend only on the spin-0 or the spin-2 degrees of freedom. This can be done with the following definition.

\begin{definition}
The \emph{spin-2 and spin-0 potentials} are defined, respectively, by
\beq
\chi_2  = \frac{\ph  + \psi }{2}
\qquad
\text{and}
\qquad
\chi_0  = 2 \psi  - \ph
.
\eeq
\end{definition}

\begin{proposition} In terms of the spin-2 and spin-0 potentials, the system of differential equations \eq{EOM-Sistema} decouples and can be expressed as a single-index equation,
\beq
\n{poi-chi}
f_s (\De) \De \chi_s = 4 \pi G \rho
, \qquad \qquad
s = 0,2,
\eeq
where 
\begin{subequations} \label{fs}
\begin{align} 
& f_2 ( \Box) = a_1 (\Box)
\,,  \n{f2}
\\
& f_0 ( \Box) = \frac{3 a_2 (\Box) - a_1 (\Box)}{2}
\,.  \n{f0}
\end{align}
\end{subequations}
\end{proposition}

\begin{proof} The Eq.~\eq{EOM-tr} can be directly rewritten in terms of $\chi_0$ and $f_0$ as~\eq{poi-chi}. For the other equation, if we take three times the Eq.~\eq{EOM-00} minus Eq.~\eq{EOM-tr} all the $a_2 (\De)$-dependent terms cancel, and we get 
\beq 
a_1 (\De) \De (\ph + \psi) 
= 8 \pi G \rho
\,, 
\eeq
which is the same as~\eq{poi-chi} using the definitions for $\chi_2$ and $f_2$.
\qed
\end{proof}

Once the solutions for  
$\chi_{0,2}$ are obtained, the original potentials $\ph$ and $\psi$ can be recovered as a linear combination of them through
the inverse transformation
\beq
\label{pch}
\ph =  \frac43 \chi_2 - \frac13 \chi_0
, \qquad  \qquad
\psi = \frac23 \chi_2 + \frac13 \chi_0
.  
\eeq
In this sense, one can work with the spin-$s$ potentials without loss of generality.

\begin{observation} \label{Obs3}
\normalfont 
\begin{svgraybox}
The spin-$s$ potentials were introduced in \cite{BreTib1} and applied in the subsequent works~\cite{BreTib2,BuoGia,Nos4der,Nos6der}. From the technical point of view, they facilitate the considerations 
since the simplicity of Eq.~\eq{poi-chi} (in comparison to the system~\eq{EOM-Sistema} of coupled equations) makes it possible to derive general results based on particular characteristics of the functions $f_s$. For instance, this decomposition is critical to applying the effective source formalism to a general higher-derivative gravity model in the Newtonian limit.

The spin-$s$ potentials also have a clear physical interpretation.  As the name suggests, these auxiliary potentials 
separate the contributions of the gauge-invariant scalar and spin-2 degrees of freedom of $h_{\mu\nu}$.
To understand this, we need to remember the spin-2 and spin-0 projection operators for symmetric rank-2 tensors~\cite{Barnes,Rivers}, namely,
\beq
&&
P^{(2)}_{\mu\nu,\al\be} = \frac12 
\left( 
\th_{\mu\al} \th_{\nu\be} + \th_{\mu\be} \th_{\nu\al} 
\right)
- \frac{1}{3} \, \th_{\mu\nu} \th_{\al\be}
,
\\
&&
P^{(0-s)}_{\mu\nu,\al\be} = 
\frac{1}{3} \, \th_{\mu\nu} \th_{\al\be}
,
\eeq
where the transverse and longitudinal vector projection operators are, respectively, 
\beq
\n{vecprodef}
\th_{\mu\nu} = \eta_{\mu\nu} - \om_{\mu\nu} 
,
\qquad
\om_{\mu\nu} = \frac{\pa_\mu \pa_\nu}{\Box} 
.
\eeq
The spin-$s$ component of the field $h_{\mu\nu}$ is defined as
\beq
h^{(s)}_{\mu\nu} \equiv P^{(s)}_{\mu\nu,\al\be} \, h^{\al\be}
.
\eeq
Thus, using the above formulas and the Newtonian-limit metric~\eq{New-met}, it is possible to show that the 
gauge-invariant spin-2 and scalar components of $h_{\mu\nu}$ are determined, respectively, by the spin-2 and spin-0 potentials,
\beq
\n{hchi2}
&&
h_{\mu\nu}^{(2)} = \left[ \eta_{\mu 0} \eta_{\nu 0} - \frac13  \left( \eta_{\mu\nu} - \frac{\pa_\mu \pa_\nu}{\Box} \right) \right] \chi_2
\,,
\\
\n{hchi0}
&&
h^{(0-s)}_{\mu\nu} 
= - \frac{2}{3} \left( \eta_{\mu\nu} - \frac{\pa_\mu \pa_\nu}{\Box} \right) \chi_0
\,.
\eeq

It is also useful to recall that the gauge-independent part of the propagator associated to~\eqref{HDF} is given by
\beq
\label{prop}
G_{\mu\nu,\al\be} (k)
=
 \frac{P^{(2)}_{\mu\nu,\al\be}}{ k^2 f_2(-k^2)}
-  \frac{P^{(0-s)}_{\mu\nu,\al\be} }{ 2 k^2 f_0(-k^2)}
.
\eeq
Thus, the number of extra spin-$s$ particles in the model depends on the number of roots of the equations
$f_s(-k^2)=0$. (For example, in local higher-derivative gravity, $f_s$ is a polynomial, and the propagator can have many massive poles. On the other hand, for the nonlocal models of Eq.~\eq{Snl} the function $f_s$ is the exponential of an entire function, and there is just the massless pole of the graviton, 
at $k^2=0$---like in the case of general relativity, for which $f_0(-k^2)=f_2(-k^2)=1$.) 
In light of~\eq{poi-chi}, this also means that the potential $\chi_s$ only depends on the spin-$s$ sector of the theory. 
\end{svgraybox}
\end{observation}


\subsection{Effective delta sources}

With the previous proposition, we managed to reduce the system of linearized equations of motion to the simple form of decoupled equations for the potentials $\chi_s$. Formally, we can rewrite Eq.~\eq{poi-chi} as a standard Poisson equation,
\beq
\n{ED-Source}
\De \chi_s = 4 \pi G \rho_s
,
\eeq
with modified effective sources $\rho_s$ defined through
\beq
\label{InvFs}
\rho = f_s(\lap) \, \varrho_s.
\eeq

This procedure involves the inversion of the operator $f_s(\lap)$, which is not 
direct\footnote{Since we only consider static solutions, the original d'Alembert operator $\Box$ is substituted by the Laplacian $\lap$. This avoids all the complications related to the choice of the appropriate Green function of the inverse operator in four-dimensional space with Lorentzian signature, especially for 
models with complex poles and nonlocalities (see, \textit{e.g.},~\cite{Calcagni:2018fid,PU50} for further discussion).} 
in general and also depends
on the shape of the original source~$\rho$. Since in this section we investigate whether higher derivatives can resolve the $-1/r$ Newtonian singularity,  
here we consider the case of a static pointlike source with mass $M$, associated with a Dirac delta function,
\beq
\n{de-sour}
\rho(\vec{r}) = M \de (\vec{r}).
\eeq
In the linear limit, there is no loss of generality in considering only pointlike sources, as more complicated matter distributions can be constructed using the \emph{superposition principle}. For further examples, the reader can consult~\cite{Frolov:2015bia,Frolov:PRL,BreTib1}, where this procedure was applied to study the collapse of small-mass spherical shells of null fluid.

\begin{definition}
\label{Def-EffDS}
The \emph{effective delta source} (or, simply, \emph{effective source}) is given by 
\beq
\n{lap-eff}
\rho_s (r) = \frac{M}{2 \pi^2} \int_0^\infty \rd k \, \frac{k \sin(kr)}{r f_s(-k^2)}
,
\eeq
where $r = |\vec{r}|$.
It is important to notice that some conditions must be imposed on the function $f_s(x)$ for the integral~\eq{lap-eff} to be well defined. Namely, we assume that $f_s(-k^2)>0$ for $k\in\mathbb{R}$, $f_s(0)=1$, and that, if $f_s(-k^2)$ is not trivial, it diverges at least as fast as $k^2$ for $k \to \infty$. Although in some cases $f_s(x)$ can be a nonanalytic function (see~\cite{Nos6der} for an example), for the sake of simplicity of considerations \emph{here we always assume that $f_s(x)$ is analytic}.
\end{definition}

Under these assumptions the
Fourier kernel associated with the function $1/f_s(-k^2)$ is well defined on the space of square-integrable functions and allows one to obtain~\eq{lap-eff} through the Fourier transform method applied to~\eq{InvFs}. In fact, starting from the integral representation of the Dirac delta function,
\beq
\de(\vec{r}) =  \int \frac{\rd^3 k}{(2\pi)^3} e^{i \vec{k}\cdot \vec{r}} ,
\eeq
it follows
\beq
\rho_s = M \left[ f_s(\lap) \right] ^{-1} \de (\vec{r} ) 
=  M \int \, \frac{\rd^3 k}{(2\pi)^3} \frac{e^{i \vec{k}\cdot \vec{r}}}{f_s(-k^2)} ,
\label{primeq}
\eeq
which can be written as~\eq{lap-eff} after integrating on the angular coordinates.

The conditions imposed by the definition~\ref{Def-EffDS} to the functions $f_s(x)$ can be regarded as a restriction on the type of models we consider. 
Indeed, the requirement $f_s(0) = 1$ expresses that the theory recovers general relativity in the infrared limit. 
In the case of local models, $f_s(x)$ is a polynomial, and the condition $f_s(-k^2) > 0$  for all $k\in\mathbb{R}$ means, from the physical point of view, 
that the model does not have tachyons in the spectrum, while technically it ensures that the integrand in~\eq{lap-eff} does not have singularities on the integration interval. 
Besides that, in what concerns nonlocal models, the condition on the asymptotic behavior of $f_s(-k^2)$ on the real line acts as a constraint on the type of nonlocality of the theory. Namely, it requires that $f_s(-k^2) \sim k^2$ (or faster) for sufficiently large $k$, so that the propagator~\eq{prop} has an improved behavior in the UV with respect to general relativity (for which $f_s(-k^2)=1$). It is in this sense that the term ``nonlocal higher-derivative gravity'' should be understood throughout this chapter---in opposition to other nonlocal models that do not have this kind of improved propagator, \textit{e.g.}, the ones defined by form factors $F_s(\Box) \propto 1/\Box$ or $F_s(\Box) \propto 1/\Box^2$~\cite{Deser:2007jk,Maggiore:2016gpx}.

As a consequence of the definition of the effective sources $\rho_s$, formula~\eq{ED-Source} means that the effect 
of the higher-derivative operator $f_s(\lap)$
on the Newtonian potentials can be treated equivalently as the smearing of the original delta source. In this spirit, we have the following definition.
\begin{definition}
The \emph{mass function} $m_s(r)$  represents the effective mass inside a sphere of radius $r$,
\beq
\label{massfunction}
m_s(r) = 4\pi \int_0^r \rd x \, x^2 \rho_s (x) 
.
\eeq
\end{definition}
Accordingly, it is expected that 
\beq
\n{MInfy}
\lim_{r\to\infty} m(r) = M.
\eeq 
This can be proved from the first equality in~\eq{primeq} by recalling that $f_s(0)=1$, so that 
the first term in the series of $1/f(\lap)$ gives the original delta function (whose  integral in whole space is 1), while the other terms in the series produce derivatives of delta functions, which vanish upon integration.

With all these definitions, we can finally obtain the formal solution for the spin-$s$ potentials by rewriting~\eq{ED-Source} in spherical coordinates,
\beq
\label{ED-Source-Expandido}
\chi_s'' (r) + \frac{2}{r} \, \chi_s'(r) = 4 \pi G \, \rho_s (r)
\,.
\eeq
As it can be directly verified, we have the following result.

\begin{theorem}
\label{Theo-Sol}
The solution of Eq.~\eq{ED-Source-Expandido} for the effective delta source~\eq{lap-eff} reads
\beq
\label{chiIntg}
\chi_s (r) = - \int_\infty^r \rd x \, g_s (x) 
,
\eeq
where 
\beq
\label{gs_def}
g_s (r) = - \frac{G \, m_s(r)}{r^2}
\eeq
and $m_s(r)$ is the mass function~\eq{massfunction}.
\end{theorem}

Before we consider explicit examples, in the next section we take a closer look at this general solution and present some results related to the regularity of the Newtonian potentials and curvature invariants.

\begin{observation} \label{Obs4}
\normalfont 
\begin{svgraybox}
In the case of the proposition~\ref{a1=a2}, \textit{i.e.}, if $a_1 (\Box) = a_2 (\Box) \equiv a(\Box)$, the Eqs.~\eq{fs} 
imply
\beq
f_2 (\Box) = f_0 (\Box) = a(\Box)
,
\eeq
so that there is only one effective source 
\beq
\rho_{\rm eff}(r) \equiv \rho_2(r) = \rho_0(r)
,
\eeq
where
\beq
\n{lap-eff-eff}
\rho_{\rm eff} (r) = \frac{M}{2 \pi^2 r} \int_0^\infty \rd k \, \frac{k \sin(kr)}{ a(-k^2)}
,
\eeq
and one mass function,
\beq
\n{one-mass}
m(r) \equiv  m_2 (r) = m_0 (r) = 4 \pi \int_0^r \rd x \, x^2 \rho_{\rm eff} (x).
\eeq
Accordingly, all the potentials are equal in this situation, namely,
\beq
\n{one-pot}
\ph(r) = \psi(r) = \chi_2(r) = \chi_0(r)
.
\eeq
Given the theorem~\ref{Theo-Sol}, the solution for the potential, in this case, is 
\beq
\label{phiIntg}
\ph (r) = - \int_\infty^r \rd x \, g (x) 
,
\eeq
where 
\beq
\label{gzinho_def}
g (r) = - \frac{G \, m(r)}{r^2}
.
\eeq
\end{svgraybox}
\end{observation}


\section{Properties of the effective sources and Newtonian potentials}
\label{Sec4}

Important properties of the potentials $\chi_s(r)$ can be derived without actually specifying the functions $f_s(x)$ and solving the integrals involved in~\eq{chiIntg}, but only by knowing some basic features of $f_s(x)$ translated into the effective source $\rho_s(r)$. 
In this section, we present necessary and sufficient conditions for the Newtonian potentials to be finite at $r=0$. Then, after characterizing the models with a regular Newtonian-limit metric, we extend the discussion to the behavior of the derivatives of the potentials and the regularity of curvature invariants. 
The core elements in this consideration are the regularity properties of the effective delta sources, which will also be used to analyze the regular black holes constructed in Sec.~\ref{Sec7}.


\subsection{Regularity and higher-order regularity of the effective sources}

\begin{theorem}[\normalfont Regularization of the source~\cite{BreTib2}]
\label{TheoRegEffS} 
If there exists a $k_0>0$ such that $k>k_0$ implies that $f_s(-k^2)$ grows at least as fast as $k^{4}$, then the effective source $\rho_s(r)$ is integrable and bounded. Moreover, $\rho_s(r)$ reaches its maximum at $r=0$.
\end{theorem}
\begin{proof}
Let us define the function
\begin{equation} \label{g_def}
G_s(r,k) = \frac{k \sin(kr)}{r f_s(-k^2)},
\end{equation}
which is the integrand in Eq.~\eq{lap-eff}. Since it is assumed that $f_s(0) = 1$ and that $f_s(-k^2)$ does not change sign for $k >0$, it follows that $G_s(r,k)$ is bounded on any compact of $\mathbb{R}^2$. The integrability of~\eqref{g_def}, thus, depends on its behavior as $k \to \infty$. Under the hypothesis of the theorem,
\begin{equation}
k> k_0 \quad \Longrightarrow \quad \vert G_s(r,k) \vert  \leqslant \frac{c}{k^2}
\end{equation}
for a constant $c$. 
Using the Weierstrass test it follows that $G_s(r,k)$ is integrable on $k$, even for $r=0$, and the integral converges uniformly.
Hence, $\rho_s(r)$ is continuous, integrable, and bounded, showing that the higher derivatives regularize the $\de$-singularity of the original source.

In what concerns the maximum, since $r,k>0$ implies 
\beq
\label{187}
\vert G_s(r,k) \vert = \frac{k \vert \sin(kr) \vert}{r f_s (-k^2)} \leqslant \frac{k^2}{f_s (-k^2)} = G_s(0,k) \, ,
\eeq 
then
\beq
\int_0^\infty \rd k \, G_s(r,k)  \leqslant \int_0^\infty \rd k \, \vert G_s(r,k) \vert  \leqslant \int_0^\infty \rd k \, G_s(0,k)  ,
\eeq
which means that $r=0$ is the maximum of $\varrho_s(r)$. In particular, $\varrho_s(0) \neq  0$.
\qed
\end{proof}

As an application of this theorem, we note that if a local higher-derivative gravity model contains at least six derivatives of the metric in the spin-$s$ sector, then the related effective source $\rho_s$ is regular, regardless of whether the propagator has real or complex, simple or degenerate poles. 
On the other hand, fourth-derivative gravity does not satisfy the hypothesis of the theorem, as $f_s(-k^2)$ grows like $k^2$; indeed, this growth is not fast enough to result in a regular effective source, as we explicitly show in Sec.~\ref{Sec6.1} (see also~\cite{BreTib2}).

For the considerations in the following sections, it is necessary to investigate the behavior of the odd-order derivatives of the effective sources. To this end, let us define the \emph{order of regularity} of a regular function~\cite{Nos6der}.

\begin{definition}
Given a bounded function $\xi:[0,+\infty)\longrightarrow\mathbb{R}$ and an integer $p\geqslant 0$, we shall say that $p$ is the \emph{order of regularity} of $\xi$ 
if:
\begin{enumerate}
\item{$\xi(r)$ is at least $2p$-times differentiable on $[0,+\infty)$ and $\xi^{(2p)}(r)$ is continuous.
}
\item{If $p\geqslant 1$, the  first $p$  odd-order derivatives of $\xi(r)$ vanish as $r \to 0$, namely,
$$
0 \leqslant n \leqslant p-1 \quad \Longrightarrow \quad \lim_{r \to 0} \, \xi^{(2n+1)}(r) = 0 .
$$
}
\end{enumerate}
If these conditions are satisfied, we shall also say that the function $\xi(r)$ is \emph{$p$-regular}. 
\end{definition}

Notice that 0-regularity only means that the function is bounded, or regular in the usual sense. Moreover, 
having Taylor's theorem in mind, one can say that a function $\xi(r)$ is $p$-regular if the first $p$ odd-order coefficients of its Taylor polynomial around $r=0$ are zero. In this sense, an analytic function $\xi(r)$ is $\infty$-regular if and only if it is an \emph{even function}, \textit{i.e.}, such that $\xi(-r) = \xi(r)$. Also, the second condition of the definition means that $\xi^{(2n+1)}(r)$ vanishes at least linearly as $r\to 0$.

In terms of the definition of $p$-regularity, 
in the theorem~\ref{TheoRegEffS} we presented sufficient conditions for the effective sources to be 0-regular. In what follows, we extend this theorem for arbitrary order of regularity.

\begin{theorem}[\normalfont Higher-order regularity of the source~\cite{Nos6der}]
\label{TheoHORegEffS} 
If there exists a $k_0>0$ such that $k>k_0$ implies that $f_s(-k^2)$ grows at least as fast as $k^{4+2N}$ for an integer $N\geqslant 0$, then the effective source $\rho_s(r)$ is $N$-regular.
\end{theorem}
\begin{proof}
The case $N=0$ was already proved in theorem~\ref{TheoRegEffS}; so, hereafter we take $N\geqslant 1$.
Notice that the function $G_s(r,k)$ in~\eqref{g_def} is even in $r$, analytic, and has the series expansion
\begin{equation}
G_s(r,k) = \frac{k^2}{f_s(-k^2)} \sum_{\ell=0}^\infty \frac{(-1)^\ell}{(2\ell + 1)!}  \left( kr \right)^{2\ell},
\end{equation}
whence,
\beq \label{array0}
\lim_{r \to 0} \, \frac{\partial^{n} }{\partial r^{n}} G_s(r,k)   =  \left\{ 
\begin{array}{l l}
0 \, ,  &  \text{if } n \text{ is odd},\\
\frac{(-1)^{n/2}}{(n + 1)} \frac{k^{n+2}}{f_s(-k^2)}  \, , &  \text{if } n \text{ is even}.\\
\end{array} \right . \,
\eeq
Also, the derivatives with respect to $r$ are bounded (for a fixed $k$). Indeed, since
$$
\frac{\pa^{n+1}}{\pa r^{n+1}} G_s(r,k) = - \frac{n+1}{r} \frac{\pa^{n}}{\pa r^{n}} G_s(r,k) + \frac{k^{n+2}}{r f_s(-k^2)} \sin \bigg[ kr + \frac{(n+1)\pi}{2}\bigg] ,
$$
the extrema of $\vert\tfrac{\pa^{n}}{\pa r^{n}} G_s(r,k)\vert$ are limited by $\tfrac{k^{n+2}}{(n+1)f_s(-k^2)}$. Taking~\eqref{array0} and the analyticity of $G_s(r,k)$ in $r$ into account, we have
\begin{equation} \label{QuotaDerN}
\bigg\vert \frac{\pa^{n}}{\pa r^{n}} G_s(r,k) \bigg\vert \leqslant \frac{k^{n+2}}{(n+1)f_s(-k^2)} ,
\end{equation}
which generalizes~\eqref{187}.

Regarding the bound defined by~\eqref{QuotaDerN} as a function of $k$ it follows that if $f_s(-k^2)$ grows at least as fast as  $k^{n+4}$, then the improper integral
\beq
\label{IntUniConv}
\int_0^\infty \rd k \, \frac{\pa^n}{\pa r^n}G_s(r,k) 
\eeq
converges uniformly for $r \geqslant 0$. In these circumstances, the source $\rho_s$ is $n$-times continuously differentiable 
and we can apply differentiation under the integral sign in~\eqref{lap-eff}, namely,
\begin{equation} 
\label{rhoNint}
\rho_s^{(n)}(r) = \frac{M}{2 \pi^2} \int_{0}^{\infty} \rd k \, \frac{\pa^n}{\pa r^n}G_s(r,k) .
\end{equation}
Furthermore, in the limit $r \to 0$ the function $\tfrac{\pa^n}{\pa r^n}G_s(r,k)$, for $n$ odd, converges uniformly to zero on any compact, implying that the limit $r \to 0$ can be interchanged with the integral in~\eqref{rhoNint} (see, \textit{e.g.},~\cite{LivAn1}). Thus, the odd derivatives of the source vanish at $r=0$ because of~\eqref{array0}---the condition is that $f_s(-k^2)$ is continuous and grows at least as fast as $k^{n+4}$. To rewrite this according to the hypothesis of the theorem we use the correspondence $n \mapsto 2N$. This means that if $f_s(-k^2)$ grows at least as fast as $k^{4+2N}$ for an integer $N>0$ then, up to the $2N$-th order, all the odd-order derivatives of the effective source vanish at $r=0$, or
\beq
\lim_{r \to 0} \, \rho_s^{(2n+1)}(r) = 0, \qquad n=0,1,...,N-1.
\eeq
In other words, the effective source $\rho_s(r)$ is $N$-regular.
\qed
\end{proof}

\begin{corollary}
If $f_s(-k^2)$ asymptotically grows faster than any polynomial, then the effective source $\rho_s(r)$ is $\infty$-regular, \textit{i.e.}, it is an even and analytic function of  $r$ .
\end{corollary}

The previous theorems can be rewritten making explicit the relation between the behavior of the function $f_s(-k^2)$ for large $k$ and the number of derivatives in a gravitational action:
\begin{corollary}
If a local gravitational action has $2N+6$ derivatives of the metric in the spin-$s$ sector (for $N\geqslant 0$), then the effective source $\rho_s(r)$ is $N$-regular.
\end{corollary}

Having established these results about the regularity of the effective source, we can study the behavior of the corresponding effective mass function, defined in~\eq{massfunction}.

\begin{proposition}
\label{Prop13} 
If $f_s(-k^2)$ grows at least as fast as $k^{4+2N}$ for an $N\geqslant 0$ then, near $r=0$,
\beq
\n{limM2}
m_s(r) \, \underset{r \to 0}{\sim} \, r^3, 
\eeq
\end{proposition}
\begin{proof}
Under this hypothesis the effective source is $N$-regular, and from theorem~\ref{TheoRegEffS} it has maximum at $r=0$, so $\rho_s(0) \neq 0$. Then, from Taylor's theorem and Eq.~\eq{massfunction} we get $m(r) \approx 4\pi \rho_s(0) r^3/3$ for small enough $r$. 
\qed
\end{proof}
However, in the case of $f_s(-k^2) \sim k^2$ asymptotically, it happens that $m(r) \sim r^2$ for small $r$, 
since $\rho_s(r) \sim 1/r$ (see Sec.~\ref{Sec6.1} below and other examples in~\cite{BreTib2,Nos6der}).


\subsection{Regularity of Newtonian-limit solutions}

Returning to the Newtonian-limit solutions in higher-derivative gravity in the context of the theorem~\ref{Theo-Sol}, we can use the regularity properties of the effective source $\rho_s$ to deduce those of the potential $\chi_s$. In fact, from the above results on the small-$r$ behavior of the mass function $m_s(r)$ it follows that the function $g_s(r)$ defined in Eq.~\eq{gs_def} is bounded in any model with higher derivatives in the spin-$s$ sector,\footnote{Since the function $g_s(r)$ is related to the gravitational force exerted on a test particle, this can be interpreted in the following way: 
If $f_s(-k^2)$ grows at least as fast as $k^4$ for large $k$, the force vanishes linearly as $r\to 0$, for $g_s(r) = O(r)$; on the other hand, in the case of $f_s(-k^2) \sim k^2$ asymptotically, $g_s(0) \neq 0$ and the force is finite (but nonzero) at $r=0$.} thus the potential $\chi_s$ in Eq.~\eq{chiIntg} is also finite.

A stronger result can be obtained concerning the higher-order regularity of the potentials:
\begin{theorem}[\normalfont Higher-order regularity of the potential~\cite{Nos6der}]
\label{TheoHORegPot} 
If the effective source $\rho_s(r)$ is $N$-regular, then the potential $\chi_s(r)$ is $(N+1)$-regular.
\end{theorem}
\begin{proof}
It follows from the relation between $\chi_s$ and $\rho_s$ in terms of theorem~\ref{Theo-Sol}, the order of regularity of the effective source and  Taylor's theorem. For instance, under the hypothesis of the theorem, the first odd-order term in the Taylor expansion of the source around $r=0$ is at least of order $r^{2N+1}$, while for the potential $\chi_s(r)$ it is at least of order $r^{2N+3}$. In other words, $\chi_s(r)$ is $(N+1)$-regular.
\qed
\end{proof}

The two previous results can be combined in a theorem relating the behavior of the function $f_s(-k^2)$ in the limit of large $k$ and the order of regularity of $\chi_s(r)$.

\begin{theorem}[\normalfont Higher-order regularity of the potential~\cite{Nos6der}]
\label{TheoHORegPotII} 
If there exists a $k_0>0$ such that $k>k_0$ implies that $f_s(-k^2)$ grows at least as fast as $k^{2+2N}$ for an integer $N\geqslant 0$, then the potential $\chi_s(r)$ is $N$-regular.
\end{theorem}
\begin{proof}
The case $N=0$ was proved at the beginning of this subsection, while the case $N\geqslant 1$ follows from the combination of theorems~\ref{TheoHORegEffS} and~\ref{TheoHORegPot}.
\qed
\end{proof}

\begin{corollary}
If $f_s(-k^2)$ asymptotically grows faster than any polynomial, then the potential $\chi_s(r)$ is $\infty$-regular, \textit{i.e.}, it is an even and analytic function.
\end{corollary}

Taking into account the relation between the spin-$s$ potentials and the Newtonian-limit metric through Eqs.~\eq{pch} and~\eq{New-met}, it follows that if $\chi_s(r)$ is $N_s$-regular the potentials $\ph(r)$ and $\psi(r)$ (thus, the metric components) are $N$-regular, with $N \equiv \min \lbrace N_0, N_2 \rbrace$. In particular, for a higher-derivative gravity model to have a regular Newtonian-limit metric, it must contain higher derivatives both in the spin-2 and spin-0 sectors.  For example, in the incomplete polynomial-derivative model considered in~\cite{Quandt:1990gc}, 
with $F_2(\Box)=0$, the potentials $\ph(r)$ and $\psi(r)$ are unbounded near the origin because only $\chi_0(r)$ is regular.


\subsection{Curvature invariants in the Newtonian limit}
\label{Sec5}

As mentioned in the introduction of this chapter, the regularity of the metric components is not enough to avoid the occurrence of singularities in the curvature invariants. 
Consider, for instance, a generic Newtonian-limit metric in the form~\eq{New-met}.  The Kretschmann scalar associated with it reads
\beq
\n{Kre-lim}
(R_{\mu\nu\al\be}^2)_\text{lin} \,=\,
4 (\ph ''^2
+2\psi ''^2 )
+\frac{16}{r} \, \psi ' \psi ''
+\frac{8}{r^2}(
\ph '^2
+3 \psi '^2
)
\, ,
\eeq
where we kept only the terms of the lowest order in the metric perturbation in consonance with the Newtonian approximation. Therefore, for the Kretschmann scalar to be bounded, it suffices that the potentials  
are twice continuously differentiable, and the limits
\beq \label{RegCond}
\lim_{r \rightarrow 0} \, \chi_s^{\prime\prime}(r) 
\qquad \mbox{and} \qquad
\lim_{r \rightarrow 0} \, \frac{\chi_s^\prime(r)}{r} 
\eeq
exist for both $s=0,2$---remember that $\ph$ and $\psi$ are related to $\chi_{0,2}$ via~\eq{pch}. 

These conditions are automatically satisfied if the potentials $\chi_s$ are at least 1-regular, but they do not hold if they are only 0-regular. It is not difficult to check that the situation of the other linearized curvature invariants, such as $(R)_\text{lin}$ and $(R_{\mu\nu}^2)_\text{lin}$, is exactly the same. 
\begin{proposition}
\label{Prop14} 
All the Newtonian-limit scalars formed by the contraction of an arbitrary number of curvature tensors are bounded if the potentials $\chi_2$ and $\chi_0$ are at least 1-regular.
\end{proposition}
\begin{proof}
By definition, the Newtonian-limit scalars are evaluated at lowest order in the metric perturbation. Thus, since the curvatures are already $O(h_{\ldots})$, any contraction of indices between curvatures is performed with the flat-space metric $\eta_{\mu\nu}$. Therefore, by dimensional arguments, such scalars are formed by combinations of products of $\chi_s^{\prime\prime}$ and $\chi_s^{\prime}/r$. The desired result then comes from the fact that these functions are finite if the potential $\chi_s(r)$ is at least 1-regular.
\qed
\end{proof}

However, the potentials should be regular to a higher order for the regularity of scalars involving derivatives of the curvatures. This can be readily seen evaluating the invariant $\Box^n R$ (for an arbitrary $n$), which in the Newtonian limit reads
\beq
\label{LapN_R}
(\Box^n R)_\text{lin} & = &  2 \lap^{n+1} \chi_0
=  2 \left[ \chi_0^{(2n+2)} + \frac{2(n+1) }{r} \chi_0 ^{(2n+1)}  \right] 
.
\eeq 
Hence, to regularize the scalar $\Box^n R$ it suffices to have a potential $\chi_0(r)$ that is $(n+1)$-regular.
In addition, if $\chi_0(r)$ is analytic but $\chi_0 ^{(2n+1)}(0)\neq 0$, then $\Box^n R$ is not regular.

This simple example illustrates the relation between the higher-order regularity of the metric components and the cancellation of the singularity in scalars containing covariant derivatives of curvatures. A complete treatment of the problem at linear level was carried out in~\cite{Nos6der}, where one can find the proof of the next theorem.

\begin{definition}
Given a certain $n \in \mathbb{N}$, we denote by $\mathscr{I}_{2n}$ the set of all the scalars that are polynomial in  curvature tensors and their derivatives, with the restriction that the maximum number of covariant derivatives is $2n$. For example, $\mathscr{I}_{0}$ is the set of the scalars of type $R_{\ldots}^N$ (contraction of an arbitrary number $N$ of curvature tensors),  while $\mathscr{I}_{2}$ also contains invariants like $R\Box R$ and $(\na_\la R_{\mu\nu\al\be})^2$.
Accordingly, $\mathscr{I}_{2n} \supset \mathscr{I}_{2(n-1)} \supset \ldots \supset \mathscr{I}_{0}$. We refer to a generic element of the set $\mathscr{I}_{2n} \setminus \mathscr{I}_{0}$ as a \emph{curvature-derivative invariant}.
\end{definition}

\begin{theorem}[\normalfont Ref.~\cite{Nos6der}]
\label{TheoRegLin} 
For each $n \in \mathbb{N}$, a sufficient condition for the Newtonian-limit regularity of all the elements in $\mathscr{I}_{2n}$ is that the potentials $\chi_0$ and $\chi_2$ are $(n+1)$-regular.
\end{theorem}

Notice that proposition~\ref{Prop14} is the particular case of the theorem for $n=0$. Taking into account the relation between the behavior of the function $f_s(-k^2)$ for large $k$ and the order of regularity of the potential $\chi_s(r)$, given by theorem~\ref{TheoHORegPotII}, the previous theorem offers a characterization of the minimal higher-derivative gravity model with a regular Newtonian limit defined by curvature-derivative invariants:
\begin{corollary}[\normalfont Ref.~\cite{Nos6der}]
\label{CorCI} 
If for $k$ large enough $f_s(-k^2)$ grows at least as fast as $k^{4+2N_s}$ for an integer $N_s\geqslant 0$, then all the linearized elements in $\mathscr{I}_{2N}$ 
(where $N \equiv \min \lbrace N_0, N_2 \rbrace$)
evaluated at the Newtonian-limit metric are regular.
In other words, if a local gravitational action has $2N_s+6$ derivatives of the metric in the spin-$s$ sector, then all the Newtonian-limit scalars in the set $\mathscr{I}_{2N}$ are regular. In the nonlocal gravity models which $f_{0,2}(-k^2)$ asymptotically grows faster than any polynomial, all the curvature-derivative invariants are regular.
\end{corollary}


\section{Selected examples}
\label{Sec6}

Here we show some explicit examples in which the general results presented in the previous sections can be verified. We focus on three interesting higher-derivative models, \textit{viz}. 
fourth-derivative gravity, polynomial-derivative gravity, and exponential nonlocal gravity. Besides the theoretical importance 
of these models (as commented in Sec.~\ref{Sec2}), they serve well 
to illustrate the increasing order of regularity of the solutions as more derivatives are included in the action.

\subsection{Fourth-derivative gravity}
\label{Sec6.1}

The solution for the Newtonian potential in fourth-derivative gravity was obtained by Stelle in 1977~\cite{Stelle77,Stelle78}. Let us explain how to reproduce this result in the formalism presented above.

The model is described by the action \eq{4HD} with real dimensionless coefficients $\al_{1,2,3}$, thus the form factors $F_{1,2}$ in Eq.~\eq{HDF} are constants, namely,
\beq
F_1 = 16\pi G \left( \al_3 - \al_1 \right) 
\,,
\qquad 
F_2 = 16\pi G \left( 4 \al_1 + \al_2 \right) 
\,.
\eeq
These formulas follow from the comparison of Eqs.~\eq{4HD} and \eq{HDF2}, and applying~\eq{F-rel}. 
Therefore, using~\eq{a1a2f} we find for the spin-$s$ functions~\eq{fs},
\beq
\n{f-Stelle}
f_s(-k^2) = 1 + c_s k^2
\,,
\eeq
where 
\beq
\n{cs-Stelle}
c_2 = - 16\pi G (4 \al_1 + \al_2) 
\,,
\qquad 
c_0 =  32 \pi G (\al_1 + \al_2 + 3 \al_3)
\,.
\eeq
As discussed in the observation~\ref{Obs3}, the roots of the equation $f_s(-k^2)=0$ are related to the theory's massive spin-$s$ degrees of freedom. Since we want to consider the general case with higher derivatives in both spin-0 and spin-2 sectors, we assume $c_{0,2} \neq 0$ and rewrite~\eq{f-Stelle} in the more convenient form
\beq
\n{fSt}
f_s (-k^2) = \frac{k^2 + \mu_s^2}{\mu_s^2} 
\,,
\eeq
where $k^2 = -\mu_s^2\,$ is the root of $f_s (-k^2) = 0$, \textit{i.e.},
\beq
\mu_s^2 \equiv  \frac{1}{c_s} .
\eeq
In order to avoid tachyons on the model and guarantee the existence of asymptotically flat solutions, we must assume $\mu_s^2 > 0$, or $c_s > 0$. Of course, in view of Eq.~\eq{cs-Stelle}, this imposes some constraints on the parameters of the action. Note that this assumption is also in consonance with the requirement that the function $f_s(-k^2)$ does not change sign for $k \in\mathbb{R}$ (see definition~\ref{Def-EffDS}).

Using~\eq{lap-eff} and~\eq{fSt} we obtain the effective delta source,
\beq
\n{lap-eff-Stelle}
\rho_s (r) = \frac{M \mu_s^2}{2 \pi^2 r} \int_0^\infty \rd k \, \frac{k\sin(kr)}{k^2 + \mu_s^2}
= \frac{M \mu_s^2}{4 \pi r} \, e^{-\mu_s r}
.
\eeq
The integral in~\eq{lap-eff-Stelle} can be easily evaluated using Cauchy's residue theorem, noting that the integrand (as a complex function) has simple poles at $k = \pm i \mu_s$. 

Now, inserting~\eq{lap-eff-Stelle} into~\eq{massfunction} gives the result for the effective mass function,
\beq
\n{mR}
m_s (r) = M \, \left[ 1- (1 + \mu_s r) e^{-\mu_s r} \right] 
\eeq
from which we obtain the function $g_s(r)$ in~\eq{gs_def} and, finally, the spin-$s$ potential~\eq{chiIntg},
\beq
\n{chisR}
\chi_s (r) = - \frac{G M}{r} \left( 1 - e^{-\mu_s r} \right) 
.
\eeq
Substituting this expression into~\eq{pch}, the Newtonian potentials $\ph(r)$ and $\psi(r)$ follow as a linear combination of $\chi_{0,2}$, namely,
\beq
\n{np}
\ph (r) & = & - \frac{GM}{r} \Big( 1 - \frac43 e^{\mu_2 r} + \frac13 e^{\mu_0 r} \Big), 
\\
\n{nf}
\psi (r) & = & - \frac{GM}{r} \Big( 1 - \frac23 e^{\mu_2 r} - \frac13 e^{\mu_0 r} \Big).
\eeq

The small-$r$ behavior of the functions $\rho_s(r)$, $m_s(r)$, $g_s(r)$, and $\chi_s(r)$ can be easily obtained from the previous formulas,
\begin{align}
\nonumber
\rho_s(r) =  \frac{M \mu_s^2}{4\pi} \Big( \frac{1}{r} - \mu_s \Big) + O(r), \quad\quad\quad   &  m_s(r) = \frac{M \mu_s^2}{2} r^2 + O(r^3) , 
\\
g_s(r) = - \frac{G M \mu_s^2}{2} + O(r), \quad \quad \quad \,\,\,\,\quad\quad  &
\chi_s(r) = - G M \mu_s \left( 1 - \frac{\mu_s}{2} r \right)  + O(r^3)
.
\nonumber
\end{align}
Thus, we see that $\rho_s$ diverges as $1/r$ and the mass function goes like $r^2$ for small $r$, the spin-$s$ potential is finite at $r=0$ and the force (proportional to $g_s$) tends to a nonzero constant. All these properties 
agree with the general results of Sec.~\ref{Sec4} for a model with a function $f_s(-k^2)$ that grows like $k^2$.

Even though the Newtonian $1/r$ singularity in the potentials $\chi_s$ is regularized, they are not 1-regular, as $\chi_s'(0) = G M \mu_s^2/2 \neq 0$. Therefore, the metric has a curvature singularity at $r=0$; indeed, using Eqs.~\eq{Kre-lim},~\eq{np} and \eq{nf}, we find that near the origin the Kretschmann scalar behaves like
\beq
(R_{\mu\nu\al\be}^2)_\text{lin} \underset{r \to 0}{\sim}  \frac{8 G^2 M^2}{9 r^2} (\mu_0^4 + \mu_0^2 \mu_2^2 + 7 \mu_2^4)
.
\eeq
This divergence, though, is less strong than in general relativity, for which $(R_{\mu\nu\al\be}^2)_\text{lin} \sim r^{-6}$.


\subsection{Polynomial higher-derivative gravity}
\label{Sec6.2}

The case of polynomial-derivative gravity involves the action~\eq{4HD} enlarged by the higher-order terms~\eq{ALS}, which contain $2N+2$ derivatives of the metric. This is equivalent to have the action~\eq{HDF} with form factors $F_{1,2} (z)$ that are polynomials of maximum degree $N-1$, for a certain $N \geqslant 2$. 
Note that we excluded the possibility of $N=1$, for which the polynomials are trivial, and the model becomes the fourth-derivative gravity considered in the previous subsection. 
Under these conditions, let us consider that $f_s(z)$ is a real polynomial of degree $N_s \geqslant 2$ and, to simplify the analysis, that the equation $f_s(z) = 0$ has $N_s$ distinct roots. (The more general case with degenerate roots can be found in~\cite{BreTib1,BreTib2}.) However, we do not restrict the occurrence of complex conjugate pairs of roots, as in the case of Lee-Wick gravity~\cite{ModestoShapiro16,Modesto16}.

Therefore, recalling that $f_s(0)=1$, the function $f_s(z)$ can be factored as
\beq
\n{fz-poly}
f_s(z) = \prod_{i=1}^{N_s} \frac{\mu_{s,i}^2 - z}{\mu_{s,i}^2} .
\eeq
To solve the integral in~\eq{lap-eff} we can apply the partial fraction decomposition,
\beq \label{PartFrac1}
\frac{1}{f_s(z)} = \prod_{i=1}^{N_s} \frac{\mu_{s,i}^2}{\mu_{s,i}^2 - z}  = \sum_{i=1}^{N_s} C_{s,i} \frac{ \mu_{s,i}^2}{\mu_{s,i}^2 - z}  ,
\eeq
where  
\beq
\n{coi}
C_{s,i} =  \prod_{\substack{j=1\\j \neq i}}^{N_s} \frac{\mu_{s,j}^2}{\mu_{s,j}^2 - \mu_{s,i}^2}  .
\eeq
Thus,
\beq
\n{rho-Int-pa}
\rho_s (r) = \frac{M}{2 \pi^2 r} \sum_{i=1}^{N_s} C_{s,i} \, \mu_{s,i}^2 \int_0^\infty \rd k \, \frac{k\sin(kr)}{k^2 + \mu_{s,i}^2} .
\eeq
Since each integral in~\eq{rho-Int-pa} is the same as the one in~\eq{lap-eff-Stelle}, \textit{mutatis mutandis}, the result is
\beq
\n{lap-eff-poly}
\rho_s (r) 
= \frac{M}{4 \pi r}  \sum_{i=1}^{N_s} C_{s,i} \, \mu_{s,i}^2 \,\, e^{-\mu_{s,i} \, r}
.
\eeq
As consequence,
\beq
\n{m-eff-poly}
m_s (r) = M \, \sum_{i=1}^{N_s} C_{s,i} [1- (1 + \mu_{s,i} r) e^{-\mu_{s,i} r} ]
\eeq
and
\beq
\n{poly-class}
\chi_s (r) = - \frac{GM}{r} \, \Bigg(  1
- \sum_{i=1}^{N_s} C_{s,i} \, e^{-\mu_{s,i} \, r}
\, \Bigg).
\eeq

Some general comments about this solution are in order. 
Even if some of the quantities $\mu_{s,i}^2$ are complex, like in the case of Lee-Wick gravity, the potential \eq{poly-class} is a real function. 
The main reason for this is that the function $f_s(z)$ is a real polynomial with real coefficients, so the integral involved in the calculation of the effective source~\eq{lap-eff} cannot result in a complex function.
Therefore, if the coefficients $C_{s,i}$ and some of the functions in the above formulas take complex values, their imaginary parts are canceled in the combinations present in~\eq{lap-eff-poly}, \eq{m-eff-poly} and \eq{poly-class}---an explicit proof of this statement can be found in~\cite{Giacchini:2016xns}. 
This cancellation happens because, from the fundamental theorem of algebra, the complex masses always occur in complex conjugate pairs. In particular, we notice that
\beq
\Re \left( C_{s,i} \, e^{-\mu_s r} \right)  = \left[ c_\text{R} \,  \cos(\mu_\text{I} r) \, + \, c_\text{I} \, \sin(\mu_\text{I} r)\right] e^{-\mu_\text{R} r} ,
\eeq
where we denote $\mu_s = \mu_\text{R} + i \mu_\text{I}$ and $C_{s,i} = c_\text{R} + i c_\text{I}$. Hence, 
the presence of complex poles in the propagator can result in oscillatory contributions to the effective source, the mass function, and the potentials~\cite{Modesto16,Accioly:2016qeb,Giacchini:2016xns,Jens,Bambi-LWBH}.

The explicit verification of the order of regularity of the effective source and the potential is based on this lemma, whose proof can be found in~\cite{Nos6der} (see also~\cite{PU50}).
\begin{lemma}
Given the polynomial~\eq{fz-poly} with $N_s$ distinct roots $\mu_{s,i}^2$ ($i=1,\ldots,N_s$), the following $N_s+1$ identities are true,
\beq \label{array1}
\sum_{i=1}^{N_s} C_{s,i} \, (\mu_{s,i})^{2n}   =  \left\{ 
\begin{array}{l l}
1  										    , & \quad \text{if } \, n =0,\\
0 										    , & \quad \text{if } \, n =1,\ldots, N_s-1,\\
(-1)^{N_s-1} \prod_{i=1}^{N_s}  \mu_{s,i}^2  , & \quad \text{if } \, n = N_s,
\end{array} \right .
\eeq
with $C_{s,i}$ defined in~\eq{coi}.
\end{lemma}

\begin{proposition}[\normalfont Higher-order regularity in polynomial-derivative gravity~\cite{Nos6der}]
\label{PropHORPDG} 
The effective source~\eq{lap-eff-poly} is $(N_s-2)$-regular and the potential~\eq{poly-class} is $(N_s-1)$-regular. Moreover, they cannot be regular to an order higher than that.
\end{proposition}
\begin{proof}
Once proving that the above functions $\rh_s(r)$ and $\chi_s(r)$ are analytic, we only have to show that the first odd-order coefficients of their Taylor series around $r=0$ are null. Both results can be easily obtained by expanding the exponential functions in power series. Indeed, from Eq.~\eq{lap-eff-poly} we get
\beq
\rh_s(r) = \frac{M}{4\pi} \sum_{\ell=0}^\infty \frac{(-1)^{\ell}}{\ell!} \Bigg[ \sum_{i=1}^{N_s} C_{s,i} \, (\mu_{s,i})^{\ell+2} \Bigg] r^{\ell-1} .
\eeq
The result now follows from the application of the above lemma, and for each $n = 1,\ldots, N_s-1$ the formula~\eq{array1} shows that $\rh_s(r)$ is $(n-1)$-regular. In particular, Eq.~\eq{array1} with $n=1$ shows that the effective source is bounded as $r$ approaches 0 (and, thus, it is bounded everywhere). Whereas Eq.~\eq{array1} with $n=N_s$ shows that the coefficient of the term of order $r^{2N_s-3}$ is nonzero, being proportional to the product of all the quantities $\mu_{s,i}^2$. Thus, $\rh_s(r)$ is $(N_s-2)$-regular, but it is not $(N_s-1)$-regular.

The same reasoning can be applied to the potential $\chi_s(r)$ in~\eq{poly-class},
\beq 
\label{Taylor-Pot}
\hspace{-2mm}
\chi_s (r) =  - G M \, \Bigg[  \Bigg( 1 - \sum_{i=1}^{N_s} C_{s,i} \Bigg)  \, \frac{1}{r} 
+ \sum_{k=0}^\infty \frac{(-1)^{k+1}}{(k+1)!} \, \Bigg( \sum_{i=1}^{N_s} C_{s,i} \, \mu_{s,i}^{k+1} \Bigg) \, r^k \Bigg] .
\eeq
The 0-regularity of the potential comes from Eq.~\eq{array1} with $n=0$, while the other $N_s$ identities show that $\chi_s(r)$ is $(N_s-1)$-regular, but it is not $N_s$-regular.
\qed
\end{proof}

The above result provides the explicit verification of the theorems~\ref{TheoHORegEffS} and~\ref{TheoHORegPotII} in the context of polynomial-derivative gravity. Also, from the Taylor expansion of the effective mass function~\eq{m-eff-poly} around $r=0$ it is straightforward to obtain
\beq
m_s (r) = -\frac{M}{3} \sum_{i=1}^{N_s} C_{s,i} \mu_{s,i}^{3} r^3 + O (r^4) ,
\eeq
in agreement with the proposition~\ref{Prop13}. Finally, it is worth noticing that the statement of the last theorem can only hold if $N_s \geqslant 2$, a condition we assumed at the beginning of this subsection. This marks the  
critical difference between theories defined by actions with four and more than four metric derivatives in what concerns the regularization of the effective delta source, the 1-regularity of the potentials, and the cancellation of the curvature singularities.

The latter can be explicitly visualized 
by calculating 
some curvature invariants, which also serves as an example of theorem~\ref{TheoRegLin}. Let us consider the particular case of sixth-derivative gravity with $\mu_{0,i} = \mu_{2,i} \equiv \mu_{i}$ ($i=1,2$). Since $N_2=N_0=2$, from proposition~\ref{PropHORPDG} we know that $\chi_{0,2}$ are 1-regular but they are not 2-regular. Using Eqs.~\eq{pch}, \eq{Kre-lim}, \eq{coi} and~\eq{poly-class} we can evaluate the linearized Kretschmann scalar near $r=0$,
\beq
(R_{\mu\nu\al\be}^2)_\text{lin}  = \frac{20 (G M)^2 \mu_{1}^4 \mu_{2}^4}{3 (\mu_{1}+\mu_{2})^2} + O(r),
\eeq
which is 
finite, as the potentials $\chi_{0,2}(r)$ are 1-regular. However, since they are not 2-regular, we expect that there are singular scalars with two covariant derivatives; for instance, from~\eq{LapN_R} we have
\beq
(\Box R)_\text{lin}  = - \frac{2 G M \mu_{1}^2 \mu_{2}^2}{r} + O(r^0)
\eeq
in agreement with theorem~\ref{TheoRegLin}.

\begin{observation} \label{Obs5}
\normalfont 
\begin{svgraybox} 
We close this example with a brief digression about the literature on Newtonian-limit solutions in polynomial higher-derivative gravity. 
In 1991,~Quandt and~Schmidt~\cite{Quandt:1990gc} published the result for the potential in the incomplete polynomial model, with higher derivatives only in the spin-0 sector, {\it i.e}, with $F_2(\Box) = 0$ in \eq{HDF}. 
Because the spin-2 sector of the model does not have higher derivatives, the potentials found still diverge like $-1/r$. 
The Newtonian potential for the complete polynomial model with real simple poles in the propagator was obtained in~\cite{Modesto:2014eta}, where it was noted that the potential is finite at $r=0$. 
In~\cite{Giacchini:2016xns} it was shown that the cancellation of the singularity in the potential happens in any local model with higher derivatives in the spin-2 and spin-0 sectors, regardless of whether the number of derivatives is the same or the nature of the poles (real or complex, simple or degenerate). 
Explicit expressions for the potentials in the case of complex poles were first published in~\cite{Modesto16,Accioly:2016qeb}; some applications of solutions with complex and degenerate real poles were studied in~\cite{Accioly:2016qeb,Accioly:2016etf}.
The complete expression for the potential in a general polynomial gravity theory, including complex poles and poles with any multiplicity, was obtained in~\cite{BreTib1}.  Also, the decomposition of the potential in terms of its spin-parts was introduced, and it was proved that all models with at least six derivatives in both spin-2 and spin-0 sectors have a regular Newtonian limit without curvature singularities.    
The formulation of this problem in terms of effective delta sources was carried out in~\cite{BreTib2}, aiming at 
extending  
the result to the case of nonlocal gravity models. 
Finally, in~\cite{Nos6der} the idea of higher-order regularity 
was introduced, and it was shown the relation between the number of derivatives in the action, the order of regularity of the potential, and the regularity of linearized curvature-derivative invariants.
\end{svgraybox}
\end{observation}


\subsection{Nonlocal gravity}
\label{non-local section}
\label{Sec6.3}

The last example we consider belongs to the class of nonlocal theories defined by the action~\eq{Snl}. As discussed in the observations~\ref{Obs2} and~\ref{Obs4}, in this case, there is only one independent effective source $\rho_{\rm eff}(r)$ and potential $\ph(r)$, inasmuch as the spin-2 and spin-0 parts of the propagator have exactly the same behavior, related to the function [compare Eqs.~\eq{Snl} and~\eq{act-lm}]
\beq
\n{non-local ff}
a (\Box) = e^{H(\Box)}
.
\eeq
There is a huge arbitrariness in the choice of the entire function 
$H(z)$, and the order of regularity of the effective delta source can be affected by this choice---in fact, as shown in Sec.~\ref{Sec4}, it depends on the behavior of $a(-k^2)$ for large $k$. Since the previous examples 
dealt with functions that behave like polynomials, here we focus on a model for which $a(-k^2)$ grows faster than any polynomial as $k\to\infty$. To this end, we choose the simplest entire function,\footnote{For the evaluation of the potential with more complicated functions, see, \textit{e.g.},~\cite{Frolov:2015usa,Edholm:2016hbt,Boos:2018bxf} and, in particular,~\cite{Jens,BreTib2,Nos6der}, for considerations regarding the effective source.} $ H(z) = - z/\mu^2$, so that
\beq
\n{non-local fff}
a (-k^2) = e^{k^2/\mu^2},
\eeq
being $\mu>0$ a parameter with mass dimension, known as the nonlocality scale.

In this case, the effective delta source obtained from~\eq{lap-eff-eff} has a Gaussian profile,
\beq
\n{gauss}
\rho_{\rm eff} (r) = \frac{M}{2 \pi^2 r} \int_0^\infty \rd k \, \frac{k \sin(kr)}{ e^{k^2/\mu^2}} = \frac{M \mu^3  }{8 \pi ^{3/2}} \, e^{-\frac{1}{4} \mu ^2 r^2},
\eeq
and the mass function~\eq{one-mass} is given by
\beq
\n{non-l-mass}
m(r) = M \left[ \text{erf}\left(\frac{\mu  r}{2}\right) - \frac{\mu r }{\sqrt{\pi }} \,  e^{-\frac{1}{4} \mu ^2 r^2}
\right]
,
\eeq
where
\beq
\text{erf} (x) = \frac{2}{\sqrt{\pi}} \int_0^x \rd t \, e^{-t^2}
\eeq
is the error function. 
Finally, we have the potential
\beq
\n{Tseytlin}
\ph (r) = -\frac{G M }{r} \, \text{erf}\left(\frac{\mu  r}{2}\right)
.
\eeq
The solution~\eq{Tseytlin} was first derived by Tseytlin~\cite{Tseytlin95} in the framework of string theory, 
using effective sources, and it was later obtained in the context of nonlocal gravity in~\cite{Modesto12,Biswas:2011ar}.

We observe that, in agreement with the results of Sec.~\ref{Sec4}, the source $\rho_{\rm eff}(r)$ and the potential are even analytic functions of $r$. Indeed, using the series expansion of the error function, we get
\beq
\ph (r) =   -\frac{G M \mu }{\sqrt{\pi}}   \sum_{n=0}^\infty \frac{(-1)^n}{(2n+1)n!} \left(\frac{\mu  r}{2}\right)^{2n} .
\eeq
Therefore, $\rho_{\rm eff}(r)$ and $\ph (r)$ are $\infty$-regular because $a(-k^2)$ grows faster than any polynomial. This ensures the regularity of \emph{all} the curvature and curvature-derivative local invariants (see corollary~\ref{CorCI}). For example, using~\eq{Tseytlin} we get for the Kretschmann scalar~\eq{Kre-lim},
\beq
(R_{\mu\nu\al\be}^2)_\text{lin}  = \frac{5 G^2 M^2 \mu ^6 }{3 \pi } + O(r).
\eeq
Moreover, it is easy to verify that $m(r)$ is an odd function, and its behavior is the one expected from 
proposition~\ref{Prop13}.


\begin{figure}[t]
\begin{center}
\sidecaption
\includegraphics[scale=0.85]{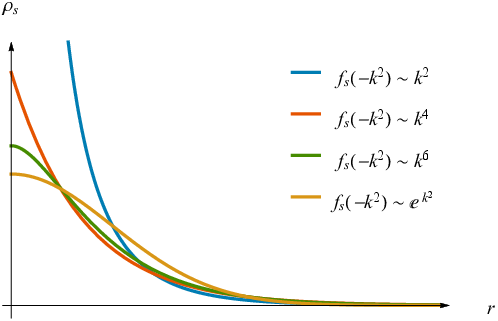}
\caption{Qualitative behavior of the effective sources for models with four (blue, $1/r$ divergence), six (red, 0-regular), and eight (green, 1-regular) derivatives, and for the nonlocal exponential model (yellow, $\infty$-regular). Notice the increasing order of regularity with the number of metric derivatives in the action. For the last two sources $\rho_s'(0)=0$, but differences appear in the  higher derivatives of $\rho_s$.
}
\label{fig:S}     
\end{center}
\end{figure}

Finally, we close the section of examples with a graph illustrating the increasing regularity of the effective delta source as the behavior of the function $f_s(-k^2)$ 
is improved. In Fig.~\ref{fig:S} we compare the effective sources for models with four, six, and eight metric derivatives, and also the nonlocal model discussed above. From the graph it is clear that
in the case of fourth-derivative gravity the source is still singular [see Eq.~\eq{lap-eff-Stelle}], while it is regular in all the other cases. However, $\rho_s'(0)\neq 0$ for the sixth-derivative model, indicating that the source is not 1-regular. For the model with eight derivatives and the exponential nonlocal model, $\rho_s'(0) = 0$, confirming the 1-regularity of the sources. The qualitative behavior of the graphs of these two sources is very similar; nonetheless, the differences appear when higher derivatives are evaluated, as $\infty$-regularity cannot be achieved in polynomial models.


\section{Regular black holes from effective delta sources}
\label{Sec7}

In the previous sections, we developed the formalism of effective sources. First, we showed how, in the Newtonian limit, the effect of the higher derivatives can be treated as the smearing of the original delta source. Subsequently, we derived the properties of the effective source depending on the behavior of the form factor and related them with the regularity of the curvature invariants. Now we shall apply the results about the effective sources to 
construct regular black hole metrics. 

To this end, let us consider a generic static and spherically symmetric metric with line element written in the convenient form
\beq
\label{metric}
\rd s^2 = - A(r) e^{B(r)} \rd t^2 + \frac{\rd r^2}{A(r)} + r^2 \rd \Omega^2
,
\eeq
where $A(r)$ and $B(r)$ are two arbitrary functions to be determined.
The main idea is to obtain solutions (in the above form) of the Einstein's equations
\beq
\label{nlEE2}
G^{\mu} {}_{\nu} = 8 \pi G \, \tilde{T}^{\mu} {}_{\nu}
\eeq
sourced by the effective energy-momentum tensor
\beq
\n{effT}
\tilde{T}^{\mu} {}_\nu = \text{diag}(-\rho_{\rm eff} , p_r , p_\th , p_\th ) ,
\eeq
where $\rho_{\rm eff}$ is the effective source given by Eq.~\eq{lap-eff-eff}. 
Notice that $\tilde{T}^{\mu} {}_\nu$ includes not only the effective delta source, but also some nontrivial effective radial ($p_r$) and tangential ($p_\th$) pressures. These components are essential for the validity of the conservation equation $\nabla_\mu \tilde{T}^\mu {}_\nu = 0$ and, therefore, for the consistency of the system~\eq{nlEE2}. 

In fact, for the metric \eq{metric}, the nonzero components of the Einstein tensor are
\beq
\n{Gtt}
&& 
G^t {}_t  =  \frac{A'}{r}+\frac{A-1}{r^2} =  G^r {}_r - \frac{A B'}{r} 
,
\\
\n{Gang}
&&
G^\th {}_\th   =  G^\phi {}_\phi  =  \, \frac{A'}{r} + \frac{1}{2}\frac{A B'}{r} + \frac{3}{4} A' B' + \frac{1}{4} A B'^2  + \frac{1}{2} A^{\prime\prime} + \frac{1}{2} A B^{\prime\prime}
.
\eeq
Hence, Eq.~\eq{nlEE2} with~\eq{effT} is a system of three equations to be solved for the two functions $A(r)$ and $B(r)$. A moment's reflection shows that the $tt$ and $rr$ components of~\eq{nlEE2} can be rearranged in the form
\begin{subequations} \label{EqG}
\begin{align} 
& \n{EqGtt}
\frac{A'}{r}+\frac{A-1}{r^2} = - 8\pi G \rho_{\rm eff}
\,, \\
& \n{EqGrr}
\frac{A B'}{r} = 8\pi G ( \rho_{\rm eff} + p_r )
\,,
\end{align}
\end{subequations}
in which the equation for $A(r)$ is decoupled. So, the solution for $B(r)$ can be obtained from Eq.~\eq{EqGrr}
after $A(r)$ is determined by Eq.~\eq{EqGtt},
and the remaining equation (for the component $\th\th$) is actually a consistency relation. The latter is equivalent to the  
conservation of $\tilde{T}^{\mu} {}_\nu$ and can be cast in the form
\beq
\n{Conserva}
p_r^\prime = - \frac{1}{2} \left( \frac{A^\prime}{A} + B^\prime \right)  \left( p_r + \rho_{\rm eff} \right)  - \frac{2}{r} \left( p_r - p_\th \right).
\eeq
Therefore, the field equations~\eq{nlEE2} are equivalent to the system formed by~\eq{EqG} and the conservation equation~\eq{Conserva}.

The solution for $A(r)$ can be directly obtained from Eq.~\eq{EqGtt}, as it only depends on the effective source. It is not difficult to verify that
\beq
\n{asol}
A(r) = 1 - \frac{2G m(r)}{r}
,
\eeq
where $m(r)$ is the same mass function defined in~\eq{one-mass}. Furthermore, taking into account the relation between $m(r)$ and the function $\ph(r)$, given by Eq.~\eq{phiIntg}, one can express Eq.~\eq{asol} as
\beq
\n{asol-ph}
A(r) = 1 - 2 r \ph^\prime(r)
.
\eeq
Thus, the Newtonian-limit solutions obtained in the previous sections of this chapter can be used to construct solutions of Eq.~\eq{nlEE2}.

In order to solve the remaining equations, 
it is necessary to specify the pressure components of the effective energy-momentum tensor. Here, 
we adopt the equation of state for the radial pressure,
\beq
\n{EOS}
p_r(r) = \left[  A(r) - 1 \right] \rho_{\rm eff}(r),
\eeq
as a means to generate regular black hole solutions.
Now we can solve Eq.~\eq{EqGrr} for $B(r)$, and determine $p_\th$ through
the conservation equation~\eq{Conserva}---indeed, using~\eq{EOS} it reduces to
\beq
\n{Pth-A}
p_\th = \frac{1}{4} \left[3 A^\prime r + (4 + B^\prime  r ) A - 4 \right] \rho_{\rm eff} + \frac{A - 1}{2} \, r \rho_{\rm eff}^\prime.
\eeq
Let us mention that there are other possible choices for the pressure components that can lead to regular black hole solutions, 
but due to the space limitation,
we shall only explicitly discuss the ones originated from the equation of state~\eq{EOS}. A slightly more general result, though, is presented in the theorem~\ref{TheoHORegBGen} below.

Having found the solution~\eq{asol} [or, equivalently,~\eq{asol-ph}] for $A(r)$ and specified the pressure $p_r$, the function $B(r)$ follows from~\eq{EqGrr}, namely,
\beq
\n{Bsol}
B(r) =  8\pi G \int_{\infty}^r \rd x \, \frac{x \left[ \rho_{\rm eff} (x) + p_r (x) \right]}{A(x)}
 = 8 \pi G \int_\infty^r \rd x \, x \, \rho_{\rm eff} (x) 
,
\eeq
where we used~\eq{EOS}.
Therefore, we find for the metric~\eq{metric},
\beq
\label{metricB=0}
\hspace{-1cm}
\rd s^2 &=& - \left(1-\frac{2G m(r)}{r} \right) e^{  \, 8 \pi G \, \int_\infty^r \rd x \, x \, \rho_{\rm eff} (x)   } \, \rd t^2 
+ \frac{\rd r^2 }{{\left(1-\frac{2G m(r)}{r} \right)}} 
+ r^2 \rd \Omega^2
.
\eeq

\begin{observation} \label{Obs6}
\normalfont 
\begin{svgraybox}
Even though~\eq{nlEE2} is not the field equation of any higher-derivative model
described previously in this chapter, sometimes it is  used 
in an effective approach, 
as an approximation of the
field equations of the models described by the action~\eq{act-lm}. 
The reason is the specific form of the higher-derivative structure in~\eq{act-lm}, which makes it possible to factor the operator $a(\Box)$ together with the Einstein tensor in the field equations. Indeed, taking the variational  derivative of~\eq{act-lm} with respect to the metric, we find
\beq
\n{cu}
a(\Box) G^{\mu} {}_{\nu} + O(R_{\ldots}^2) = 8 \pi G \, T^{\mu} {}_{\nu}
.
\eeq 
Therefore, by the inversion of the form factor $a(\Box)$ 
one might regard~\eq{nlEE2} as an approximation of~\eq{cu}, sourced by a pointlike source, in spacetime regions where the curvature is small; more details about this procedure can be found in~\cite{Modesto12,Bambi-LWBH}. 
It is worth mentioning that Eq.~\eq{nlEE2} was also used as effective equations in other frameworks---see, {\it e.g},~\cite{Isi:2013cxa,dirty,Nicolini:2005vd,Modesto:2010uh,Nicolini:2019irw,Zhang14}. 

Regardless of the physical interpretation of~\eq{nlEE2} described above, our point of view is that these equations are very interesting \textit{per se} and deserve a detailed study.
As we show below, by imposing some general requirements on the form factor $a(\Box)$, it is possible to obtain  singularity-free black holes that are  solutions of some equations of motion.
\end{svgraybox}
\end{observation}


\subsection{General properties of $A(r)$ and $B(r)$}

Similar to the case of the Newtonian limit, by just considering the asymptotic behavior of the form factor $a(z)$, translated into the effective source $\rho_{\rm eff}(r)$, we can explain many important physical properties of the solutions.

The properties of $A(r)$, for instance, can be derived straightforwardly 
from the relation between it and the function $\ph(r)$, via~\eq{asol-ph}.
\begin{theorem}
\label{TheoHORegA} 
If the effective source $\rho_{\rm eff}(r)$ is $N$-regular for an integer $N\geqslant 0$ (or, in an equivalent way, if the function $a(-k^2)$ asymptotically grows at least as fast as $k^{4+2N}$), then the 
function $A(r)$ 
is $(N+1)$-regular. Moreover,
\beq
\n{ALim}
\lim_{r \to 0} \, A(r) \, = \, \lim_{r \to \infty} \, A(r) \, = \, 1.
\eeq 
\end{theorem}
\begin{proof}
The first part is 
a direct consequence of theorems~\ref{TheoHORegPot} and~\ref{TheoHORegPotII}, which imply that $r\ph^\prime(r)$ is $(N+1)$-regular. Also, since $\ph(r)$ is at least 1-regular, near $r=0$ we have $r\ph^\prime(r) = O(r^2)$ whence, using~\eq{asol-ph}, $A(0) = 1$. The remaining limit follows from~\eq{MInfy}.
\qed
\end{proof}
As a corollary, if $a(-k^2)$ asymptotically grows faster than any polynomial, $A(r)$ is an even function.

Under the hypotheses of theorem~\ref{TheoHORegA} for the effective source $\rho_{\rm eff}(r)$,
and using the result~\eq{ALim} for $A(r)$, it can be shown that 
the effective pressure components defined by Eqs.~\eq{EOS} and~\eq{Pth-A} have the following features:
\begin{enumerate}
\item{They vanish asymptotically for large and small $r$.}
\item{They are finite at the horizons, \textit{i.e.}, in the spacetime regions where $A(r) = 0$.}
\end{enumerate}
Other choices of effective pressures can have this same qualitative behavior, 
or less stringent ones 
(\textit{e.g.}, being only finite, but nonvanishing, at $r=0$), see the discussion after theorem~\ref{TheoHORegBGen}, below.

Another important consequence of Eq.~\eq{ALim}, and the fact that $A(r)$ is bounded for all the effective delta sources considered above,\footnote{Notice that~\eq{ALim} is valid also for $a(-k^2) \sim k^2$, which is not covered by theorem~\ref{TheoHORegA}. Indeed, in this case $\ph(r)$ and $A(r)$ are 0-regular, but not 1-regular; yet, $A(0)=1$ because $r\ph^\prime(r) = O(r)$.} is the existence of a critical mass $M_\text{c}$ such that $M < M_\text{c}$ implies $A(r) > 0$ for all $r$. This means that for this source there exists a \emph{mass gap} for a solution to describe a black hole. 
Indeed, as $A(r)$ does not change sign if $M < M_{\rm c}$, the metric does not have any horizon. On the other hand, if $M > M_{\rm c}$, the solution has an even number of horizons because the function $A(r)$ changes sign an even number of times. 
This is in contrast  with the case of the delta source in general relativity, for which black hole solutions exist regardless of the value of $M$. Further discussion on the mass gap for black hole solutions in higher-derivative gravity can be found in~\cite{Frolov:PRL,Frolov:2015usa,Frolov:2015bia,BreTib1,Frolov:Weyl}.

In what concerns the function $B(r)$ in~\eq{Bsol}, its regularity properties can be directly deduced from the results on the order of regularity of the effective source $\rho_{\rm eff}(r)$, stated in theorem~\ref{TheoHORegEffS}.
\begin{theorem}
\label{TheoHORegB} 
If the effective source $\rho_{\rm eff}(r)$ is $N$-regular for an integer $N\geqslant 0$ (or, in an equivalent way, if the function $a(-k^2)$ asymptotically grows at least as fast as $k^{4+2N}$), then the 
function $B(r)$ in~\eq{Bsol}
is $(N+1)$-regular.
\end{theorem}

A more general theorem, valid for a broader family of equations of state for $p_r$ can be formulated as follows.
\begin{theorem}
\label{TheoHORegBGen} 
Let $\si(x)$ be a continuous, $N$-regular function (for an integer $N\geqslant 0$), such that the function $\xi(x) = x \si(x)$ is integrable on $[0,+\infty)$.
If we assume the equation of state
\beq
p_r(r) = A(r) \si(r) - \rho_{\rm eff}(r) ,
\eeq
then $B(r)$ is $(N+1)$-regular and is given by 
\beq
\n{BsolGen}
B(r)  = 8 \pi G \int_\infty^r \rd x \, x \, \si (x) 
. 
\eeq
\end{theorem}
The proof is immediate after applying Taylor's theorem. Note that the equation of state~\eq{EOS} and the theorem~\ref{TheoHORegB} correspond to the particular choice $\si(r) = \rho_{\rm eff}(r)$. 

Also, if $\si(r)$ is only 0-regular, then $B(r)$ is only 1-regular. 
This happens, \textit{e.g.}, for the choice $\si(r) \propto r \rho_{\rm eff}(r)$ 
where $\rho_{\rm eff}(r)$ is an effective delta source of any order of regularity 
(see~\cite{dirty} for an explicit example). In this case $B(r)-B(0) \propto m(r)$, 
which is not 2-regular, as it is immediate from Eq.~\eq{limM2}. Therefore, for 
this choice of $\si(r)$, the Kretschmann scalar and the other curvature-squared 
invariants are bounded, but invariants containing at least two covariant 
derivatives, such as $\Box R$, might diverge (see theorem~\ref{PropBoxR} below).


\subsection{Curvature regularity of the solutions}

The theorem~\ref{TheoRegLin} presented in Sec.~\ref{Sec5} related the regularity order of the components of the metric~\eq{New-met} and the regularity of sets of linearized curvature-derivative invariants. In particular, it was shown that the occurrence of odd powers of $r$ could make some curvature invariants 
singular at $r=0$. A similar situation happens in the nonlinear regime. Since the complete analysis of the problem exceeds the scope of this text, we shall only present some results in this direction.

To this end, let us consider a generic static and spherically symmetric metric with line element~\eq{metric},
where $A(r)$ and $B(r)$ are two arbitrary regular analytic functions. 
Thus, they can be written in the form of power series,
\beq
\n{ps}
A (r) = \sum_{\ell = 0}^\infty a_\ell r^\ell
,
\quad \quad
B (r) = \sum_{\ell = 0}^\infty b_\ell r^\ell
.
\eeq 
By direct calculation it is straightforward to verify that the conditions for the regularity of the invariants $R$, $R_{\mu\nu}^2$ and $R_{\mu\nu\al\be}^2$ are
\begin{subequations} \label{condregmet}
\begin{align} 
& a_0 = 1 , \label{condregmet-1} \\ 
& a_1 = b_1 = 0 . \label{condregmet-2}
\end{align}
\end{subequations}
For example, by substituting the series~\eq{ps} into the expression for the Kretschmann scalar, we get
\beq
&& R_{\mu\nu\al\be}^2 
\, = \,
4 \, \tfrac{1+ a_0(a_0-2)}{r^4}
+8 \, \tfrac{ a_1 (a_0 - 1)}{r^3}
+2 \, \tfrac{2 a_1 (a_0 b_1 + 2 a_1) + 4 a_2( a_0 -1) + a_0^2 b_1^2}{r^2}
\nonumber
\\
&& \, + 4 \, \tfrac{
 a_1( a_0 b_1^2 + 2 a_0  b_2 +  a_1 b_1 + 6 a_2 )
+2 b_1 ( a_0^2 b_2 +  a_0 a_2) + 2 (a_0 -1) a_3   }{r}
+  O\left(r^0\right) .
\n{Kre1}
\eeq 
Therefore, $R_{\mu\nu\al\be}^2$ is finite if the conditions~\eq{condregmet} are satisfied.

As proved in theorems~\ref{TheoHORegA} and~\ref{TheoHORegB}, if the effective source $\rho_{\rm eff}(r)$ is regular, the functions $A(r)$ in~\eq{asol} and $B(r)$ in~\eq{Bsol} are at least 1-regular, thus $a_1=b_1=0$; while the requirement $a_0= A(0) = 1$ is automatically satisfied by the solution in~\eq{asol}, see~\eq{ALim}. We conclude that for the curvature-squared invariants to be bounded, with the equation of state~\eq{EOS}, it suffices that $\rho_{\rm eff}(r)$ is regular.

The case of the invariants of the type $\Box^N R$ was considered in detail in Ref.~\cite{Giacchini:2021pmr}, where the following result was proved.
\begin{theorem}
\label{PropBoxR}
Let $N_A \geqslant 1$ and $N_B \geqslant 1$ be, respectively, the order of regularity of the functions $A(r)$ and $B(r)$ in the metric~\eq{metric}, which also satisfy~\eq{condregmet}, and let $N \equiv \min \lbrace N_A, N_B \rbrace$. Then, for each $n\in\mathbb{N}$ such that $n \leqslant N-1$, the scalar $\Box^{n} R$ is finite at $r=0$, whereas the invariant $\Box^{N} R$ might be singular.
\end{theorem}

The extension of this proposition to general curvature-derivative invariants is a more involved task, but it is expected that one similar to theorem~\ref{TheoRegLin} also holds in the nonlinear regime. In particular, all the local curvature-derivative scalars are finite at $r=0$ if the metric~\eq{metric} is an even and analytic function of $r$. Further considerations and examples can be found in Ref.~\cite{Giacchini:2021pmr}.


\subsection{Example of regular black hole: The case of nonlocal form factor}

As an explicit example of the procedure outlined in this chapter to construct regular black hole metrics, here we consider the nonlocal form factor $a(-k^2)=\exp(k^2/\mu^2)$. 
This case, in analogy with Sec.~\ref{Sec6.3}, displays an 
interesting feature: Since the form factor $a(-k^2)$ grows faster than any polynomial for large enough $k$, the metric is an even analytic function, and, as a consequence, all the curvature invariants are regular.

According to the discussion along this section, to obtain the explicit form of the metric~\eq{metricB=0} the only new ingredient is the function $B(r)$, as the functions $\rho_{\rm eff } (r)$ and $m(r)$ for this form factor
were already evaluated in~\eq{gauss} and~\eq{non-l-mass}, respectively. Thus, substituting the latter into~\eq{asol} we find 
\beq
\label{Azao}
A(r) = 1 -\frac{2 G M }{r} \, \text{erf}\left(\frac{\mu  r}{2}\right)
+\frac{2 G M \mu }{\sqrt{\pi }} \, e^{-\frac{1}{4} \mu ^2 r^2} ,
\eeq
and using~\eq{gauss} and~\eq{Bsol}, we obtain 
\beq
\n{BsolFFF}
B(r) =  \frac{ G M \mu^3  }{ \sqrt{\pi} } \int_\infty^r \rd x \, x \,  e^{-\frac{1}{4} \mu ^2 x^2}
= -\frac{2 G \mu  M }{\sqrt{\pi }} \, e^{-\frac{1}{4} \mu ^2 r^2},
\eeq
so that the metric~\eq{metricB=0}, in this case, is given by
\beq
\label{metricB=0-exp}
\rd s^2 &= &- \left[1 -\frac{2 G M }{r} \, \text{erf}\left(\frac{\mu  r}{2}\right)
+\frac{2 G M \mu }{\sqrt{\pi }} \, e^{-\frac{1}{4} \mu ^2 r^2} \right] \exp \left( -\frac{2 G \mu  M }{\sqrt{\pi }} \, e^{-\frac{1}{4} \mu ^2 r^2}  \right) \rd t^2 
\nonumber
\\
&&
+ \left[ 1 -\frac{2 G M }{r} \, \text{erf}\left(\frac{\mu  r}{2}\right)
+\frac{2 G M \mu }{\sqrt{\pi }} \, e^{-\frac{1}{4} \mu ^2 r^2} \right]^{-1} \rd r^2 + r^2 \rd \Omega^2
.
\eeq

As mentioned before, this metric describes a black hole if the mass $M$ is bigger than the critical mass $M_{\rm c}$ for the black hole formation. This is related to the minima of the function $A(r)$. 
In the specific case of~\eq{Azao} there is only one absolute minimum $A_{\rm min}$. If  $A_{\rm min} <0$ the function~\eq{Azao} flips sign twice and the equation $A(r) = 0$ has two roots, which we denote by $r_\pm$ (with $r_+ > r_-$). The values $r_{+}$ and $r_-$ are, respectively, the positions of the event horizon and an inner horizon.\footnote{In addition, the position of the event horizon is bounded by the Schwarzschild radius, $r_+ \leqslant r_{\rm s} = 2GM$, because in this example $m(r) \leqslant M$.} Otherwise, if $A_{\rm min} > 0 $ the function $A(r)$ does not change sign and the metric~\eq{metricB=0-exp} does not describe a black hole. These features are shown in Fig.~\ref{fig:8}, where we plot the graph of $A(r)$ for the two distinct scenarios.

\begin{figure}[t]
\begin{center}
\sidecaption
\includegraphics[scale=.85]{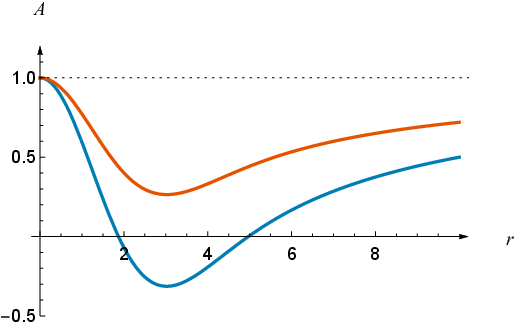}
\caption{Plot of $A(r)$ 
for $\mu = M_{\rm P} = 1$ in two situations. In blue line, with $M= 2.5 > M_{\rm c} = 1.9$, so that  $A_{\rm min} < 0$ and the equation $A(r) = 0$ has two solutions, {\it i.e}, we have two horizons at $r_- = 1.87$ and $r_+ = 4.67 < r_{\rm s} = 5.0$. In red line, with  $M = 1.4 < M_{\rm c}$, so that $A_{\rm min} > 0$ and the metric has no horizon.}
\label{fig:8}     
\end{center}
\end{figure}
 
To determine the critical mass, it is useful to rewrite Eq.~\eq{Azao} in terms of a function $q(x)$ that only depends on the dimensionless combination $\mu r$, namely,
\beq
A(r) =  1 - 2GM \mu \, q(\mu r),
\eeq 
where
\beq
\label{q}
q(x) = \frac{1}{x}\, \text{erf}\left(\frac{x}{2}\right)-\frac{ e^{-\frac{x^2}{4}} }{\sqrt{\pi }}
.
\eeq
As expected from theorem~\ref{TheoHORegA}, 
$q(x)$ is a continuously differentiable even function, it is bounded and satisfies $\lim_{x \to 0} q(x) = \lim_{x \to \infty} q(x) = 0$, in agreement with~\eq{ALim}. In fact, the graph of \eq{q} is depicted in Fig.~\ref{fig:7}, from which we see that it has an absolute maximum, $q_{\rm max} = 0.263$ at the point $x = 3.02$. Therefore, the critical mass is obtained as the solution of the equation
\beq
A_{\rm min} = 1 -  2G M_{\rm c} \, \mu \, q_{\rm max} = 0,
\eeq
that is,
\beq
M_{\rm c} = \frac{1}{2 G \mu\, q_{\rm max}} \approx 1.9 \, \frac{M_{\rm P}^2}{\mu} ,
\eeq
where we used $G = 1/M_{\rm P}^2$ for the Planck mass.

\begin{figure}[t]
\begin{center}
\sidecaption
\includegraphics[scale=.85]{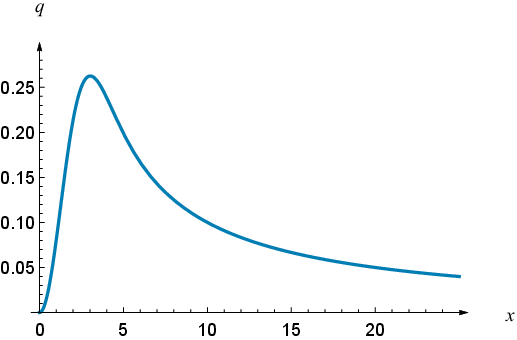}
\caption{Plot of~\eq{q}. The function $q(x)$ has an absolute maximum~$q_{\rm max} = 0.263$ at $x = 3.02$ and asymptotically vanishes for small and large values of $x$.}
\label{fig:7}      
\end{center}
\end{figure}

Regarding the regularity of the solution, since the form factor~\eq{non-local fff} grows faster than any polynomial, the metric components $A(r)$ and $B(r)$ are even functions. Therefore, we expect all curvature invariants to be finite at the origin, 
even those constructed with derivatives of the curvature tensors. 
For example, 
\beq
\label{R0}
&& R (0) = \frac{G \mu ^3 M}{\sqrt{\pi }}
,
\qquad
\Box R (0) = -\frac{G \mu ^5 M }{2 \pi } \left(2 G \mu  M+3 \sqrt{\pi }\right)
,
\\
&& 
\Box^2 R (0) =  \frac{ G \mu ^7 M }{12 \pi ^{3/2}} (45 \pi -52 G^2 \mu ^2 M^2
 +138 \sqrt{\pi } G \mu  M ),
\eeq
in agreement with the theorem~\ref{PropBoxR}. Similar relations hold for invariants constructed using the Riemann and Ricci tensors. 

It is also worth mentioning that the metric~\eq{metricB=0-exp} is consistent with the Newtonian approximation described in Sec.~\ref{Sec6.3}. Indeed, 
expanding the
temporal part of~\eq{metricB=0-exp} in powers of $GM$, 
at leading order we get
$g_{tt} = -\left[ 1 + 2 \ph(r)\right]$, with $\ph(r)$ given by~\eq{Tseytlin}.

Last but not least, although this example dealt with the case of the exponential form factor, other regular black hole solutions can be obtained from the combination of the results of Sec.~\ref{Sec4} and this section. 
In fact, within the procedure described above, provided that the function $a(-k^2)$ grows at least as fast as $k^4$ for large $k$,  
the static spherically symetric solutions of Eq.~\eq{nlEE2}  are regular in the sense that its Kretschmann scalar is  singularity-free.


\section{Concluding remarks}
\label{Conc}

In this chapter, we showed that there are situations in which the regularity of a solution of some field equations can be anticipated by knowing the regularity properties of an effective source. This happens, for instance, in the cases of the Newtonian limit and the small-curvature approximation in higher-derivative gravity 
and in other models that lead to similar effective equations, such as noncommutative gravity, generalized uncertainty principle scenarios, string theory, and other approaches for a UV completion of gravity~\cite{dirty,Nicolini:2005vd,Isi:2013cxa,Tseytlin95,Modesto:2010uh,Bambi-LWBH,Jens,Nos6der,BreTib2,Zhang14,Modesto12,Nicolini:2019irw}.

In those cases, the theorems presented in Secs.~\ref{Sec4} and~\ref{Sec7} offer an immediate answer to the question of whether a particular modification of the form factor (or, more generally, of the delta source) can lead to regular spacetime configurations and to which extent this regularity is maintained as one takes into account invariants formed by curvatures and their covariant derivatives. It should be stressed, however, that the solutions of the effective equations considered here might not describe all the properties of the original field equations of the models. This assessment can only be done case by case through the detailed analysis of the model and the assumptions underlying the approximations involved in the effective equations.


\section*{Appendix}
\addcontentsline{toc}{section}{Appendix}
\label{Appendix A}

Here we list the main formulas needed in Sec.~\ref{weak-field}, following the expansion~\eq{flat-exp}.
For the metric inverse and determinant, we have
\beq
g^{\mu\nu} & = &  \eta^{\mu\nu} - h^{\mu\nu} + h^{\mu\la} h_{\la}^{\nu} + O(h_{\ldots}^3)
,
\\
\sqrt{-g} & = & 1 + \frac12 h - \frac14 h_{\mu\nu}^2 + \frac18 h^2 + O(h_{\ldots}^3)
,
\eeq
where $h = \eta^{\mu\nu} h_{\mu\nu}$.
As explained in Sec.~\ref{weak-field}, to obtain the bilinear part in $h_{\mu\nu}$ 
of the action~\eq{HDF}
we need the Riemann and Ricci tensors only in the first order,
\beq
\n{Rie-h}
R^\al {}_{\be\mu\nu}  & = &
\frac12 \, ( 
\pa_\mu \pa_\be h^\al_\nu - \pa_\nu \pa_\be h^\al_\mu
+ \pa_\nu \pa^\al h_{\be\mu}
- \pa_\mu \pa^\al h_{\be\nu}
)
+ O(h_{\ldots}^2)
\,,
\\
\n{Ri-h}
R_{\mu\nu} & = & \frac12\, 
( \pa_\la \pa_\mu h^{\la}_\nu
+ \pa_\la \pa_\nu h^{\la}_\mu
- \Box h_{\mu\nu}
- \pa_\mu \pa_\nu h
) + O(h_{\ldots}^2)
\,,
\eeq
and the scalar curvature up to $O(h_{\ldots}^2)$, because of the term linear in $R$ in~\eq{HDF},
\beq
\n{R-h}
R =  R^{(1)} + R^{(2)}
+ O(h_{\ldots}^3)
,
\eeq
where
\beq
R^{(1)} & = & \pa_\al \pa_\be h^{\al\be} - \Box h,
\\
R^{(2)} & = & 
h^{\al\be} (\Box h_{\al\be} + \pa_\al \pa_\be h - 2 \pa_\al \pa_\la h^{\la}_\be
)
+ \frac34\, \pa_\la h_{\al\be} \pa^\la h^{\al\be} 
\nonumber
\\
&&
- \frac12\, \pa_\al h_{\la\be} \pa^\be h^{\al\la} 
- (\pa_\rho h^\rho_\la - \tfrac12 \pa_\la h) (\pa_\si h^{\si\la} - \tfrac12 \pa^\la h)
.
\eeq

With these expressions one can derive the quadratic part of the terms in the action.
Integrating by parts and ignoring unimportant surface terms, we get 
\beq
\n{R-bili} 
\left[\sqrt{-g}\, R \right]^{(2)}  =  
\frac14\, h^{\mu\nu} \Box h_{\mu\nu} 
- \frac14 h \Box h
+ \frac12 h^{\mu\nu} \pa_\mu \pa_\nu h
& 
- \frac12\, h^{\mu\nu} \pa_\mu \pa_\la h^{\la}_\nu 
,
\eeq
\beq
\n{R-exp1}
\left[ \sqrt{-g}\, R F(\Box) R \right]^{(2)} & =  &
h F(\Box)  \Box^2 h 
- 2 h^{\mu\nu} F(\Box)  \Box \pa_\mu \pa_\nu  h
\nonumber
\\
&&
+ h^{\mu\nu} F(\Box)  \pa_\mu \pa_\nu \pa_\al \pa_\be  h^{\al\be}
\,,
\\
\n{R-exp2}
\left[\sqrt{-g}\, R_{\mu\nu} F(\Box) R^{\mu\nu} \right]^{(2)} & =  &
\frac14 h_{\mu\nu} F(\Box) \Box^2 h^{\mu\nu}
+ \frac14 h F(\Box) \Box^2 h 
\nonumber
\\
&& 
- \frac12 h^{\mu\nu} F(\Box) \Box  \pa_\mu \pa_\nu  h
- \frac12 h^{\mu\nu} F(\Box) \Box  \pa_\mu \pa_\la  h^{\la}_\nu 
\nonumber
\\
&&
+ \frac12 h^{\mu\nu} F(\Box) \pa_\mu \pa_\nu \pa_\al \pa_\be  h^{\al\be}
\,,
\\
\n{R-exp3}
\left[ \sqrt{-g}\, R_{\mu\nu\al\be} F(\Box) R^{\mu\nu\al\be} \right]^{(2)} & =  & 
h_{\mu\nu} F(\Box) \Box^2 h^{\mu\nu}
- 2 h^{\mu\nu} F(\Box) \Box \pa_\mu \pa_\la  h^{\la}_\nu 
\nonumber
\\
&&
+ h^{\mu\nu} F(\Box) \pa_\mu \pa_\nu \pa_\al \pa_\be  h^{\al\be}
.
\eeq



\end{document}